\documentclass[11pt]{article}

\usepackage{amsmath,amssymb,amsthm, bbm}
\usepackage{algorithm,algpseudocode}
\usepackage{natbib}
\usepackage{enumitem}
\usepackage{graphicx}
\usepackage{setspace}
\usepackage{comment}
\usepackage{geometry}
\usepackage{hyperref}
\allowdisplaybreaks

\theoremstyle{remark}
\newtheorem{assum}{Assumption}

\theoremstyle{definition}
\newtheorem{definition}{Definition}
\newtheorem{prop}{Proposition}

\newtheorem{remark}{Remark}
\newtheorem{theorem}{Theorem}
\newtheorem*{theorem*}{Theorem}
\newtheorem{result}{Result}
\newtheorem*{result*}{Result}

\newtheorem{lemma}{Lemma}

\DeclareMathOperator*{\argmin}{arg\,min}
\newcommand{\ds}{\displaystyle}

\newcommand{\prox}{\mathop{\rm prox}\nolimits}
\newcommand{\Real}{\mathbb{R}}
\newcommand{\Tra}{^{\sf T}} 
\newcommand{\E}{\mathbb{E}}
\newcommand{\amp}{\mathop{\:\:\,}\nolimits}
\newcommand{\Var}{\operatorname{Var}} 
\newcommand{\Cov}{\operatorname{Cov}} 

\newcommand\numberthis{\addtocounter{equation}{1}\tag{\theequation}}

\title{MCMC Importance Sampling via Moreau-Yosida Envelopes}

\author{Apratim Shukla \\ Department of Mathematics and Statistics \\ IIT Kanpur \\ {\tt apratims21@iitk.ac.in} \and Dootika Vats \\ Department of Mathematics and Statistics \\ IIT Kanpur \\ {\tt dootika@iitk.ac.in} \and Eric C.\@ Chi \\ School of Statistics \\ University of Minnesota \\ {\tt  echi@umn.edu}  }

\date{\today}

\begin{document}

\maketitle

\begin{abstract}
Non-differentiable priors are standard in modern parsimonious Bayesian models. Lack of differentiability, however, precludes gradient-based Markov chain Monte Carlo (MCMC) for posterior sampling. Recently proposed proximal MCMC approaches can partially remedy this limitation by using a differentiable approximation, constructed via Moreau-Yosida (MY) envelopes, to make proposals. In this work, we build an importance sampling paradigm by using the MY envelope as an importance distribution. Leveraging properties of the envelope, we establish asymptotic normality of the importance sampling estimator with an explicit expression for the asymptotic covariance matrix. Since the MY envelope density is smooth, it is amenable to gradient-based samplers. We provide sufficient conditions for geometric ergodicity of Metropolis-adjusted Langevin and Hamiltonian Monte Carlo algorithms, sampling from this importance distribution. Our numerical studies show that the proposed scheme can yield lower variance estimators compared to existing proximal MCMC alternatives, and is effective in low and high dimensions.
\end{abstract}
\doublespacing
\section{Introduction}
\label{sec:introduction}

Markov chain Monte Carlo (MCMC) is a popular algorithm for sampling from complex and high-dimensional distributions. For a function $\psi: \Real^d \to (-\infty, \infty]$, we consider the problem of estimating characteristics of a target density of the form
\begin{eqnarray}
\label{eq:targ_dens}
    \pi(x) & \propto & e^{-\psi(x)}\,.
\end{eqnarray}
Such potentially intractable densities arise in numerous areas, but most noticeably appear as posterior distributions in Bayesian statistics. MCMC enables the estimation of features of $\pi$ by constructing a Markov chain with $\pi$ as its stationary and limiting distribution. The complexity of modern data has led to increasingly sophisticated models across various applications, and sampling from the resulting posteriors can be challenging. Consequently, much work in MCMC has gone into constructing  effective sampling and estimation strategies that are locally informed, utilizing the geometry of the posterior density to inform the next move of the Markov chain. 

Gradient-based MCMC algorithms account for the  local geometry of $\pi$  by utilizing $\nabla \psi(x)$ in the proposal distribution within  the popular Metropolis-Hastings sampler. These methods include the Metropolis-Adjusted Langevin Algorithm (MALA) \citep{roberts1998optimal}, Hamiltonian Monte Carlo (HMC) \citep{MR2858447}, Riemannian manifold MALA and HMC \citep{girolami2011riemann}, and Barker's proposals \citep{livingstone2022barker}.
 Gradient-based MCMC algorithms can converge faster to the target distribution \citep{beskos2013optimal,roberts1998optimal} than uninformed algorithms like random walk Metropolis. 
 
Chief among the requirements to implement gradient-based algorithms, is the differentiability of $\psi$. Non-differentiable priors that induce structured sparsity in model parameters are often used for modern data. Effective sampling and estimation from the resulting non-differentiable posteriors is a challenge, since it precludes a straightforward use of gradient-based MCMC methods. Major progress, however, in removing the barrier to using gradient-based MCMC methods for non-smooth target densities came in \cite{pereyra2016proximal}. This work introduced a MALA-like MCMC algorithm, called proximal MALA, for sampling from a non-smooth target density, $\pi$. The key idea is to replace gradients of $\psi$ in the MALA proposal with gradients of its Moreau-Yosida envelope, $\psi^\lambda$, which are Lipschitz. Thus,  the gradient of  $\psi^{\lambda}$ not only exists but is also well-behaved. Nonetheless, the algorithm can converge slowly, yielding high variance estimators of posterior quantities \citep{durmus2022proximal}. A primary reason of slow convergence is that the gradient of $\psi^{\lambda}$ may not adequately capture the geometry of $\psi$. Fortunately, in importance sampling, there is a natural strategy for variance reduction for exactly these situations.

Importance sampling is a classic approach of using samples generated from a proxy distribution, called the importance distribution, to estimate characteristics of a  target distribution. 
If the proxy is chosen well, the variance of the importance sampling estimator can be much smaller than the variance of the estimator computed with samples from $\pi$. Conversely, if it is chosen poorly, estimators can have larger variance -- even infinite variance. The success and failure of importance sampling hinges critically on the choice of the importance distribution.

We propose a simple and general way to design an effective importance distribution for a given $\pi$ that is either non-differentiable or lacks Lipschitz gradients. We use the Moreau-Yosida envelope, $\psi^\lambda$, to approximate the $\psi$ and implement importance sampling with the importance density $\pi^\lambda(x)\propto e^{-\psi^{\lambda} (x)}$. Since $\pi^{\lambda}$ have well-behaved gradients, gradient-based MCMC samplers like MALA and HMC  can effectively sample from $\pi^{\lambda}$. We present importance sampling estimators for expectations under $\pi$. Our estimators are guaranteed to have finite asymptotic variance. Practical ways to estimate this asymptotic variance for both univariate and multivariate expectations can be found in the supplement. As credible intervals are an integral part of the Bayesian workflow, we also provide importance sampling estimators of marginal quantiles.

Finite asymptotic variance of importance sampling estimators requires the underlying Markov chains to converge at a fast enough rate. For MALA and HMC chains, invariant for $\pi^{\lambda}$, we present sufficient conditions for geometric ergodicity. Moreover, we identify situations when MALA or HMC are not geometrically ergodic for $\pi$ but are geometrically ergodic for $\pi^{\lambda}$. 

The rest of the paper is organized as follows. Section~\ref{sec:importance_sampling} reviews importance sampling schemes and Section~\ref{sec:proximal} reviews proximal MCMC algorithms. Section~\ref{sec:MY_IS} introduces our proposed sampling scheme and estimator, discusses the estimation of quantiles, and presents critical theoretical results and practical considerations. Section~\ref{sec:geom_erg} presents results under which MALA and HMC algorithms targeting $\pi^{\lambda}$ are geometrically ergodic. Section~\ref{sec:numerical} present  numerical studies illustrating the utility of our proposed strategy over existing alternatives. 
Section~\ref{sec:discussion} ends with a discussion. All proofs and some details on the examples are provided in the supplement.
\section{Importance sampling}
\label{sec:importance_sampling} 

Consider a distribution with density $\pi$ defined on $\Real^d$, of the form in \eqref{eq:targ_dens} and a function $\xi : \Real^{d} \to \Real^{p}$. A fundamental task is to estimate
\begin{eqnarray}
   \label{eq:est_xi}
   \theta & := & \E_{\pi}\left[ \xi(X) \right] \amp = \amp \int_{\Real^{d}} \xi(x) \pi(x) dx  \amp < \amp \infty\,.
\end{eqnarray}
We use the notation $\E_{\pi}$ to indicate that the expectation is with respect to a distribution with density $\pi$. Importance sampling methods estimate $\theta$ using weighted samples from an importance distribution with density $g$, whose support contains the support of $\pi$. The key idea is to express an expectation with respect to $\pi$, as an expectation with respect to $g$,
\begin{eqnarray}
\label{eq:key}
    \theta & = & \int_{\Real^{d}} \xi(x) \pi(x) dx 
     \amp = \amp \int_{\Real^{d}} \xi(x) \dfrac{\pi(x)}{g(x)} g(x) dx \amp = \amp \E_{g} \left[\xi(X) \dfrac{\pi(X)}{g(X)} \right]\,.
\end{eqnarray}
%
%
When iid samples from $g$ are difficult to obtain, one can simulate a $g$-ergodic Markov chain $\{X_t\}_{t\geq 1}$, and a natural strategy to estimate $\theta$ in light of \eqref{eq:key}, is to use a weighted average from these samples.
%
%
When either $\pi$ or $g$ is known only up to a normalizing constant, this is not straightforward.
Fortunately, one can introduce a rescaling that eliminates the need for normalization constants. For $g(x) \propto \tilde{g}(x)$, let $w(x) = \exp(-\psi(x))/\tilde{g}(x)$. The self-normalized importance sampling estimator of $\theta$ from $g$ is
\begin{eqnarray}
\label{eq:snis_estimator}
    \hat{\theta}^{g}_n & = & \ds \sum_{t=1}^{n} \dfrac{\xi(X_t) w(X_t)}{\sum_{k=1}^{n} w(X_k)}\,.
\end{eqnarray}
When the Markov chain $\{X_t\}_{t\geq 1}$ is $g$-ergodic, $\hat{\theta}^g_n$ is strongly consistent\footnote{Theorem~\ref{thm:imp_consistency} is generally known, but we present the proof in the supplement for completeness.}.
%
\begin{theorem} 
\label{thm:imp_consistency}
   Let $\{X_{t}\}_{t\geq1}$ be an irreducible, aperiodic, and Harris recurrent Markov chain with stationary density $g$. Then, as $n \to \infty$, $\hat{\theta}^{g}_n \amp \overset{\text{a.s.}}{\to} \amp \theta$.
\end{theorem}

As a consequence of \eqref{eq:snis_estimator}, \cite{tier:1994} proposed extending importance sampling using iid sampling to MCMC sampling -- a strategy that has enjoyed some success.
\cite{silva2023robust} used MCMC importance samples for Bayesian leave-one-out cross-validation, \cite{MR1728560} obtained a univariate theoretical paradigm for problems in statistical physics, \cite{schuster2020markov} used importance sampling inspired weighted averaging to remove bias from the unadjusted Langevin algorithm, and \cite{buta2011computational,tan2015honest} used MCMC importance sampling with multiple Markov chains for applications in Bayesian sensitivity analysis. However, many of these examples are either low-dimensional or  specific for cases when natural choices of $g$ are easily available.

Importance sampling is successful when (i) it is easier to construct faster converging $g$-invariant Markov kernels than $\pi$-invariant kernels and (ii) when the tails of $g$ can adequately bound $\xi(x) \pi(x)$ to ensure low variance of $\hat{\theta}^{g}_n$. 
In fact, if $g$ is chosen well, the variance of $\hat{\theta}^{g}_n$ can be orders of magnitude lower than the variance of standard Monte Carlo. On the other hand, a finite second moment of $\xi$ under $\pi$ does not guarantee a finite variance of $\hat{\theta}^{g}_n$. If $g$ is not carefully chosen, $\hat{\theta}^g_n$ can be catastrophically worse than standard Monte Carlo. The importance density $g$ is critical to the quality of $\hat{\theta}^g_n$, but there is little general guidance on how to choose $g$ in practice. For the iid case, \citet[Chapter 2.9]{hesterberg1988advances}  provides the optimal choice of $g$ as $g^*(x) \propto |\xi(x) - \theta| \pi(x).$ 
%
%
Unfortunately, $g^*$ is not useful in practice as it depends on the unknown $\theta$.

The main contribution of this paper is a general strategy for constructing an effective $g$ for a log-concave $\pi$ that yields an estimator  $\hat{\theta}^g_n$ with finite asymptotic variance. We focus on target distributions $\pi$ when $\psi$'s proximal mapping can be computed efficiently. This enables the construction of a practical importance distribution using a smooth Moreau-Yosida approximation of $\pi$. Our proposed importance distribution is guaranteed to yield finite variance estimators. Moreover, the smoothness in the importance distribution can be tuned to yield estimators with smaller variance than standard MCMC.  While our primary focus when we started this work was on log-concave $\pi$ that are not smooth, our approach also applies to differentiable log-concave $\pi$ for which popular algorithms like MALA and HMC are ineffective. This is particularly true when $\psi$ is not Lipschitz differentiable. 

\section{Proximal Markov chain Monte Carlo}
\label{sec:proximal}

\subsection{Moreau-Yosida envelopes and proximal maps}
\label{sec:MY_env}
We review relevant concepts from convex analysis, specifically Moreau-Yosida envelopes and proximal mappings. For a thorough review of proximal mappings and their applications in statistics and machine learning, see \citet{Combettes2011, Parikh2014, PolsonScott2015Proximal}. Let $\Gamma(\Real^d)$ denote the set of proper, closed, convex functions from $\Real^d$ into $\Real \cup \{\infty\}$. Our focus is on targets $\pi$ such that $\psi \in \Gamma(\Real^d)$. Let $\| \cdot \|$ denote the  Euclidean norm.

\begin{definition}\label{def:moreau-envelope}
Given a $\lambda > 0$, the \emph{Moreau-Yosida envelope} of $\psi \in \Gamma(\Real^d)$ is given by
\begin{eqnarray}
\label{eq:psi-lambda}
\psi^\lambda(x) & = & \inf_{y \in \Real^d} \left\{\psi(y) + \frac{1}{2\lambda}\|y-x\|^2\right\}\,.
\end{eqnarray}
\end{definition}
The infimum in \eqref{eq:psi-lambda} is always attained at a unique point because $\psi \in \Gamma(\Real^d)$. The unique minimizer of \eqref{eq:psi-lambda} defines the proximal mapping of $\psi$. 
\begin{definition}\label{def:prox-operator}
Given a $\lambda > 0$, the \emph{proximal mapping} of $\psi \in \Gamma(\Real^d)$ is the operator
\begin{eqnarray}
    \label{eq:proximal_operator}
\prox_\psi^\lambda(x) & = & \underset{y \in \Real^d}{\arg\min}\; \left\{\psi(y) + \frac{1}{2\lambda}\|y-x\|^2\right\}\,.
\end{eqnarray}
\end{definition}

The function $\psi^{\lambda}(x)$ is called the Moreau-Yosida envelope of $\psi(x)$ since for all $\lambda > 0$,
\begin{eqnarray}
\label{eq:bound_psi}
    \psi^{\lambda}(x) & = &  \psi\left( \prox^{\lambda}_{\psi}(x) \right) + \frac{1}{2\lambda}  \left\|\prox^{\lambda}_{\psi}(x) - x\right\|^{2} \amp \leq \amp \psi(x) \quad \text{ for all } x \in \Real^d\,.
\end{eqnarray}
%
Therefore, $\psi^{\lambda}(x)$ ``envelopes'' $\psi(x)$ from below. 

To illustrate these concepts concretely, consider $\pi(x) \propto e^{-|x|}$, so $\psi(x) = \lvert x \rvert$. Its Moreau-Yosida envelope is the Huber function.
%
%
The left panel of Figure~\ref{fig:lap} shows $\psi(x)$ and $\psi^\lambda(x)$ for three different $\lambda$ values. We see that the Moreau-Yosida envelope provides a differentiable approximation to a non-smooth function where the approximation improves as $\lambda$ gets smaller. 

\begin{figure}
\centering
    \includegraphics[scale = 0.45]{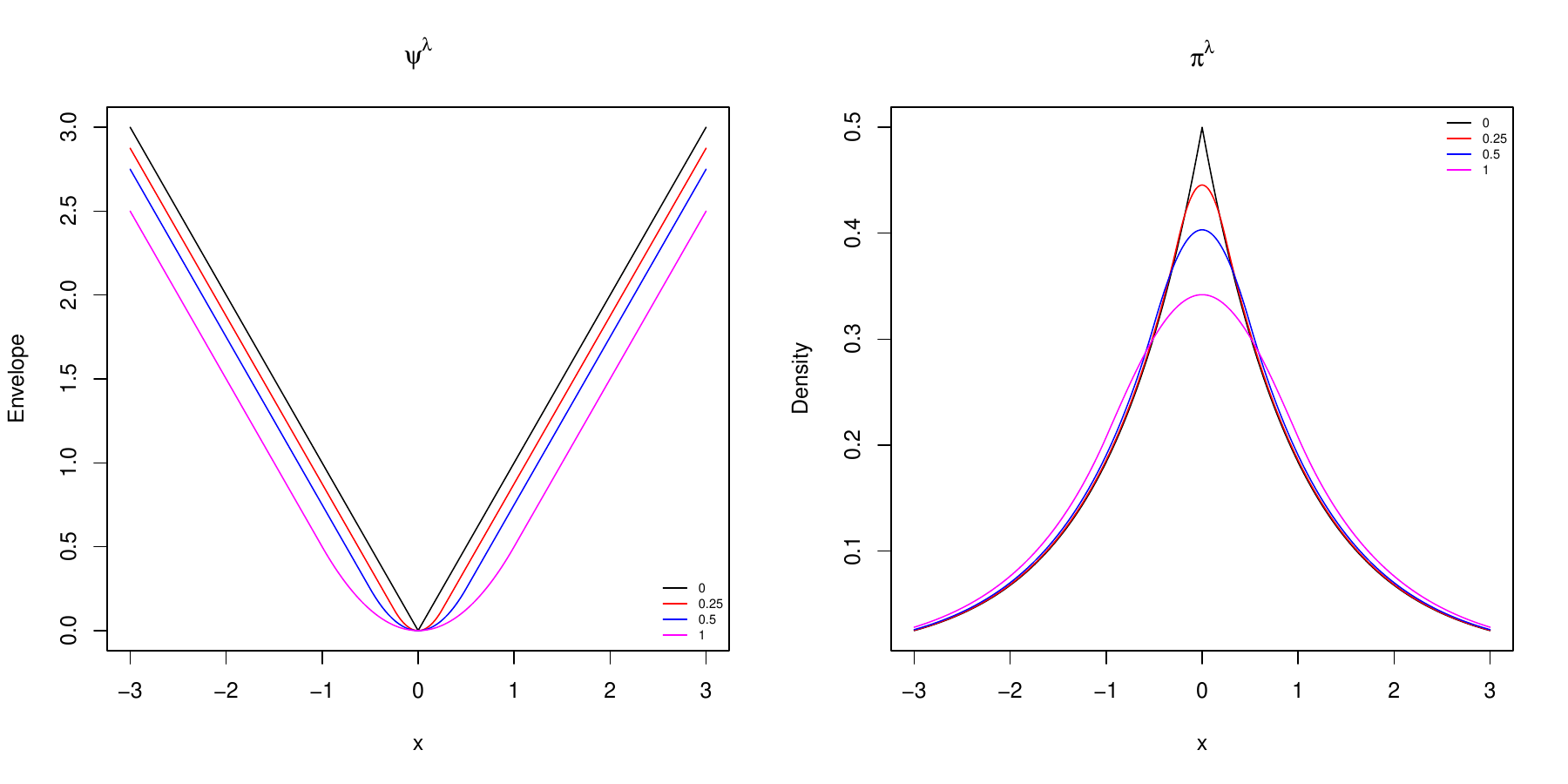}
    \caption{(Left) Moreau-Yosida envelope for $\psi(x) = |x|$  for different values of $\lambda$. (Right) The corresponding envelopes for the Laplace distribution.}
    \label{fig:lap}
\end{figure}

Having constructed $\psi^{\lambda}$, consider the importance density 
\begin{eqnarray}  
 \label{eq:targ_approx}
    \pi^{\lambda}(x) & \propto & e^{-\psi^{\lambda}(x)}\,,
\end{eqnarray}
that is integrable over $\Real^d$ for all $\lambda > 0$ \citep{durmus2022proximal}. 
The right panel of Figure~\ref{fig:lap} shows the  densities $\pi^{\lambda}$ corresponding to the Moreau-Yosida envelope, $\psi^{\lambda}$. The results in Proposition \ref{prop:my_properties} below will be critical to our proposed  importance sampling paradigm.

\begin{prop} 
\label{prop:my_properties}
[{\cite{durmus2022proximal,pereyra2016proximal}}]
Let $\psi \in \Gamma(\Real^d)$ be bounded from below. Then the following hold. 
    \begin{enumerate}[label=(\alph*)]
        \item  If $\int_{\Real^{d}} e^{-\psi(x)}\, dx < \infty\,$, then $\pi^{\lambda}$ in \eqref{eq:targ_approx} defines a proper density on $\Real^{d}$.
    \item The importance density $\pi^{\lambda}(x)$ converges to $\pi(x)$ pointwise as $\lambda \rightarrow 0$.
    \item The importance density $\pi^{\lambda}(x)$ is continuously differentiable even if $\pi$ is not. Moreover, $\nabla \log \pi^{\lambda}(x)  =   \frac{1}{\lambda}\left(\prox^{\lambda}_\psi(x) - x \right)$.
    \item  A point $x^{*} \in \Real^d$ maximizes $\pi$ if and only if  $x^{*}$ maximizes $\pi^{\lambda}$. \label{max_prop}
    \end{enumerate}
\end{prop}

\subsection{Proximal and Moreau-Yosida MCMC algorithms}

Langevin algorithms like MALA, built from a discretization of the continuous-time Langevin diffusion, are ubiquitous in modern day MCMC applications and are well studied \citep{robe:rose:2001,roberts1996exponential}. Langevin algorithms require $\log \pi(x)$ to be differentiable and its gradient cannot grow larger than $\|x\|$. Specifically, \cite{roberts1996exponential} show that for bounded $\pi$, the MALA algorithm fails to be geometrically ergodic if 
\begin{eqnarray}
\label{eq:not_geom}
    \liminf_{\|x\| \to \infty}\, \dfrac{\|\nabla \log \pi(x) \|}{\|x\|} & > & c\,,
\end{eqnarray}
where $c$ is a known expression.  Geometric ergodicity implies the existence of a Markov chain central limit theorem for ergodic averages and thus controls the overall accuracy of the results. 

Recently \cite{pereyra2016proximal} and \cite{durmus2022proximal}  extend Langevin algorithms to generate samples from non-smooth target densities. Given a target density of the form \eqref{eq:targ_dens}, the proximal MALA (P-MALA) algorithm of \cite{pereyra2016proximal} and the Moreau-Yosida unadjusted Langevin algorithm (MY-ULA) of \cite{durmus2022proximal}  employ a discretization of the Langevin diffusion for $\pi^{\lambda}$, to yield a candidate $y$ from a given $x$:
\begin{eqnarray} \label{prox_proposal}
    y & = & x + \frac{h}{2}\nabla \log \pi^{\lambda}(x) + \sqrt{h} Z\,,
\end{eqnarray}
where $Z \sim N(0, \mathbb{I}_{d})$ and $h > 0$ is a step size. \cite{pereyra2016proximal} uses $y$ as a proposal in a Metropolis-Hastings step ensuring $\pi$-invariance, and \cite{durmus2022proximal} accept $y$ as the next state of the chain, without an accept-reject step. MY-ULA is inexact due to both the discretization error and the discrepancy between $\pi$ and $\pi^{\lambda}$. P-MALA is an exact algorithm with scaling fixed as $h = 2 \lambda$. As we will see in Section~\ref{sec:numerical}, however, P-MALA can often suffer from inefficient estimation compared to our proposed solution explained in the sequel.

Hamiltonian Monte Carlo (HMC) is another popular gradient-based MCMC algorithm. \cite{chaari2016hamiltonian} construct a modified HMC proposal aimed at efficient sampling from non-smooth densities. They define the HMC proposal based on proximal mappings utilizing gradients of $\psi^{\lambda}$ instead of gradients of $\psi$. In our examples, we also compare our proposed method to this proximal mapping based HMC (P-HMC) algorithm of \cite{chaari2016hamiltonian}. 

\section{Moreau-Yosida importance sampling}
\label{sec:MY_IS}

\subsection{Moreau-Yosida importance sampling estimator}
We propose employing $\pi^{\lambda}$ as an importance density to estimate expectations and quantiles under $\pi$. 
Let $\{X_t\}_{t\geq1}$ be an irreducible, aperiodic, and Harris recurrent Markov chain with stationary density $\pi^{\lambda}$. Using notation from Section~\ref{sec:importance_sampling}, define the unnormalized weights as
\begin{eqnarray}
    \label{eq:weights}
    w^{\lambda}(x) & = & \dfrac{e^{-\psi(x)} }{ e^{-\psi^{\lambda}(x)} } \amp = \amp e^{-\left\{ \psi(x)-\psi^{\lambda}(x) \right\}}   \,.
\end{eqnarray}
To correct for the mismatch between $\pi$ and $\pi^\lambda$, our proposed \emph{Moreau-Yosida importance sampling (MY-IS)} estimator of $\theta$ takes the following weighted average of $\xi(X_t)$ 
\begin{eqnarray} \label{mym_imp_est}
    \hat{\theta}_n^{\text{MY}} & = & \sum_{t = 1}^{n} \dfrac{\xi(X_{t})w^{\lambda}(X_{t})}{\sum_{k = 1}^{n} w^{\lambda}(X_{k})} \,.
\end{eqnarray}
Computing $\hat{\theta}_n^{\text{MY}}$ is straightforward. At each iteration we just need to evaluate $\xi$ and $w^{\lambda}$ at the current Markov chain state $X_t$ and update a running sum of weights, adding nominal computational burden.

%
A universal challenge in importance sampling is ensuring the finiteness of the variance of the importance sampling estimator. The following condition on the weights ensures finite variance of an  importance sampling estimator,
\begin{eqnarray}
\label{eq:weight_bound}
    \underset{x \in \Real^d}{\sup}\; w(x) & < & \infty\,.
\end{eqnarray}
Under \eqref{eq:weight_bound} finite variance is guaranteed when $\{X_t\}_{t\geq1}$ are iid. \cite{MR1728560} present sufficient conditions for when $\{X_t\}_{t\geq 1}$ is a univariate Markov chain.

The unnormalized weights $w^\lambda(x)$ satisfy \eqref{eq:weight_bound} because of the global underestimation bound in \eqref{eq:bound_psi}. We show that \eqref{eq:weight_bound} is also sufficient to obtain asymptotic normality of $\hat{\theta}^{\text{MY}}_n$ provided the $\pi^{\lambda}$-invariant Markov chain converges at a geometric rate.

\begin{theorem}
\label{thm:asymp_norm}
    Let $\{X_t\}_{t \geq 1}$ be a $\pi^{\lambda}$-reversible, geometrically ergodic Markov chain. If  $\E_{\pi}\|\xi(X_1)\|^{2} < \infty$, then as $n \to \infty$
    \begin{eqnarray}        
    \sqrt{n}\left(\hat{\theta}_{n}^{\text{MY}} - \theta\right) & \overset{d}{\to} & N_{p}(0, \Xi)\,, \,\text{ where } \,\,
    \Xi   =  \dfrac{1}{ \left[\E_{\pi^{\lambda}}(w^{\lambda}(X_1)) \right]^{2}}
     \begin{bmatrix}
        \textbf{I}_{p} & -\theta
    \end{bmatrix} \Sigma \begin{bmatrix}
        \textbf{I}_{p} \\
        -\theta\Tra
    \end{bmatrix}\,, \label{eq:true_lambda}
\end{eqnarray}

%
and
 \begin{equation}
 \label{eq:Sigma}
 \begin{split}
    \Sigma &  \amp = \amp   \sum_{k=-\infty}^{\infty}\Cov_{\pi^{\lambda}}(S(X_{1}), S(X_{1+k})) \,\,\, \text{ with } \,\,\,
    S(x)  \amp = \amp  \begin{pmatrix}
        \xi(x)w^{\lambda}(x) \\
        w^{\lambda}(x)
    \end{pmatrix}.
\end{split}
\end{equation}
\end{theorem}

\begin{proof}
    See the supplement.
\end{proof}

\begin{remark}
    Using the results of \cite{jone:2004}, sufficient conditions can also be obtained when the $\pi^{\lambda}$-Markov chain is non-reversible, uniformly ergodic, or polynomially ergodic.   
\end{remark}

The utility of Theorem~\ref{thm:asymp_norm} is three-fold: (i) it guarantees that all choices of $\lambda$ yield a finite variance estimator, (ii) moment conditions are under $\pi$ and not the relatively unstudied $\pi^{\lambda}$, and (iii) the expression of the limiting variance $\Xi$ is explicit. In the supplement, we provide estimators of $\Xi$ which enable practitioners to assess  simulation quality and determine whether sufficiently many Monte Carlo samples have been obtained. See \cite{agarwal2022principled,glyn:whit:1991,roy:2019,vats2019multivariate}. 

Theorem~\ref{thm:asymp_norm} requires the Markov chain for $\pi^{\lambda}$ to be reversible and geometrically ergodic. In Section~\ref{sec:geom_erg} we present sufficient conditions for MALA and the HMC  algorithms to be geometrically ergodic. The details of the $\pi^{\lambda}$ algorithms are provided in the supplement.

\subsection{Quantile estimation}
\label{sec:quantile_est}

Credible intervals are critical for Bayesian analysis. \cite{glynn1996importance} present importance sampling quantiles and establish fundamental theoretical guarantees when $\pi$ is one-dimensional. Since Bayesian application areas are typically high-dimensional, their methodology does not apply directly. We instead employ the importance sampling quantile estimation procedure proposed by \cite{chen1999monte} since it can be applied in high-dimensional applications.

Let the target and importance densities be defined as in \eqref{eq:targ_dens} and \eqref{eq:targ_approx}. For $x \in \Real^d$, let $x_i$ denote its $i^{\text{th}}$ component. Further, let $\pi_i$ denote the $i^{\text{th}}$ marginal density of $\pi$ and $\Pi_i$ denote its cumulative distribution function. The $\alpha$\textsuperscript{th} marginal quantile for component $i$
\begin{eqnarray*}
    x_{i}^{(\alpha)} & = & \inf\left\{y \in \Real: \Pi_{i}(y) \geq \alpha\right\}\,.
\end{eqnarray*}
Let $\mathbbm{1}_{\mathcal{A}_s}
(\cdot)$ denote the indicator function on the set $\mathcal{A}_s = \{t \in \Real : t \leq s\}$. Then
\begin{eqnarray*}
    \Pi_{i}(s) \amp = \amp \int_{-\infty}^{s} \pi_{i}(t)\, dt \amp = \amp \E_{\pi_{i}} \left(\mathbbm{1}_{\mathcal{A}_s}(T)  \right),
\end{eqnarray*}
where $T \sim \pi_{i}$. Taking $\xi(x) = \mathbbm{1}_{\mathcal{A}_s}(x)$ in \eqref{mym_imp_est} and using Theorem~\ref{thm:imp_consistency} produces a consistent estimator of the marginal distribution functions at any given point. Let $X^{i}_t$ denote the $i^{\text{th}}$ component of the Markov chain iteration $X_t$. For $s \in \Real$, an estimator of $\Pi_i(s)$ is
\begin{eqnarray} 
\label{eq:general_emp_cdf}
    \widehat{\Pi}_{i}(s) & = & \dfrac{\sum_{t = 1}^{n}\mathbbm{1}_{\mathcal{A}_s}(X^i_{t}) w^{\lambda}(X_t)   }{ \sum_{k = 1}^{n}   w^{\lambda}(X_{k})}\,.
\end{eqnarray}
A quantile estimator is typically obtained by inverting the empirical distribution function using order statistics. This is not  straightforward for unnormalized densities, and \cite{chen:shao:1998} propose the following. 
Consider the ordered sample $\left\{X_{i,(l)}\right\}$, where $X_{i,(l)}$ is the $l$\textsuperscript{th} order statistic obtained by sorting the entire $d-$tuple according to values in the $i$\textsuperscript{th} component. Denote
\begin{eqnarray} 
\label{eq:relative_wts}
    w_{(l)}^i & = &  \dfrac{ w^{\lambda} \left(X_{i,(l)} \right) }{\sum_{t = 1}^{n}  w^{\lambda} \left(X_{i,(t)}  \right)}\,,
\end{eqnarray}
as the relative weights corresponding to the ordered sample observations for the $i$\textsuperscript{th} component. 
An importance sampling estimator of the $\alpha$\textsuperscript{th} quantile for component $i$ is,
\begin{eqnarray}
\label{eq:impsamp_quantile}
    \hat{x}_{i}^{(\alpha)} & = & \begin{cases}
        X_{i,(1)} \quad &\text{if} \quad \alpha = 0 \\
        X_{i,(m)}  \quad &\text{if} \quad \sum_{l=1}^{m - 1}w_{(l)}^i < \alpha \leq \sum_{l=1}^{m}w_{(l)}^i\,.
    \end{cases}
\end{eqnarray}
The estimator in \eqref{eq:impsamp_quantile} can be cycled over all components for the desired quantiles to produce a complete set of credible intervals. \cite{chen1999monte} showed that if $\pi$ is unimodal and $\{X_t\}_{t \geq 1}$ is a $\pi^{\lambda}$-ergodic Markov chain, then $\hat{x}_{i}^{(\alpha)}$ produces consistent credible intervals. Asymptotic normality of the estimators, however, has not been studied for when $\{X_t\}_{t \geq 1}$ is obtained via MCMC samples, and remains an open problem.

\subsection{Tuning}
\label{sec:choosing_lambda}

The MY-IS estimator requires  choosing a $\lambda$ and Metropolis-Hastings tuning parameters. The parameter $\lambda$ impacts performance of the estimator via the importance distribution while both impact the convergence of the $\pi^{\lambda}$-Markov chain. The latter is somewhat easier to address due to  abundant guidelines available. We employ MALA and HMC algorithms for $\pi^{\lambda}$ and follow the recommendations of \cite{roberts1998optimal} to tune MALA to attain approximately $57\%$ acceptance and \cite{beskos2013optimal} to tune HMC to attain $65\%$ acceptance, respectively. The above is implemented after a value of $\lambda$ is chosen, fixing the MCMC target distribution. The choice for $\lambda$ is a little more subtle as we describe below.

We choose $\lambda$ with the aim to minimize the asymptotic variance in \eqref{eq:true_lambda}. 
Given a possibly non-smooth target density $\pi$, the importance density $\pi^{\lambda}$ is smoother, encouraging faster mixing of the underlying MALA or HMC. Further, as we will motivate below, $\pi^{\lambda}$ is expected to express reduced correlation across components, which in-turn facilitates sampling, leading to the variances in  $\Sigma$ being small. Therefore, larger values of $\lambda$ will be more desirable from a sampling perspective.  On the other hand, larger values of  $\lambda$ will increase the discrepancy between $\psi(x)$ and $\psi^{\lambda}(x)$ and thus decrease $\E_{\pi^{\lambda}}(w^{\lambda}(X))$. This leads to increase in the variance,  $\Xi$. There is thus a tension between choosing a large and small value of $\lambda$. 

In the iid sampling scenario, there are guidelines for choosing the importance distribution based on the effective sample size\footnote{Later, we will refer to another effective sample size in the context of MCMC samples. To avoid confusion, we denote the effective sample size in importance sampling as $n_e$.}  of \cite{kong1992note} in importance sampling. \cite{kong1992note} estimates the effective sample size of $n$ samples from $\pi^{\lambda}$ as
\begin{eqnarray} 
\label{eq:imp_samp_ess}
    n_{\text{e}} & = & n \,\dfrac{\bar{w}_n^{2}}{\,\,\overline{w^{2}}_n}\,,
\end{eqnarray}
where $\overline{w}_n = n^{-1} \sum_{t=1}^{n} w^{\lambda}(X_{t})$ and $\overline{w^{2}}_n = n^{-1} \sum_{t=1}^{n} (w^{\lambda}(X_{t}))^{2}$. In iid sampling, an ideal choice of a proposal yields a reasonably high $n_e/n$. In our context, this suggests using a small $\lambda$ but this would not necessarily yield any benefits in MCMC importance sampling. Similarly, a value of $n_e/n$ close to 0 implies that the first moment of the weights is very small relative to the second moment, again unideal. Empirically, we find that choosing a value of $\lambda$ such that $n_{e}/n \in [0.4, 0.8]$
%
%
typically balances the tradeoff between high and low values of $\lambda$.

A practitioner may start with an initial choice of $\lambda_0$ and adjust it until the weights yield $n_e/n$ in the window above. A natural question is how to choose an initial $\lambda_0$. This is challenging to answer in general, but fixing $\pi$ to be the density of a Gaussian distribution, the following theorem indicates that the ideal $\lambda$ is inversely proportional to the dimension. Let $\Omega$ be a $d \times d$ positive-definite matrix and let $|\cdot|$ denote the determinant. 
\begin{theorem}
\label{thm:my_normal}
    Let $\pi$ be the density of $N(0, \Omega)$. Then, the MY envelope density is the density of the $N(0, \Omega + \lambda \mathbb{I}_d)$ distribution. Further, let $s_1 \leq s_2 \leq \cdots \leq s_d$ be the eigenvalues of $\Omega$. Then for iid samples from $\pi^{\lambda}$, the value of $\lambda$ corresponding to $\xi(x) = x$ that minimizes $\left| \lim_{n\to \infty} n  \text{Var} \left(\hat{\theta}^{\text{MY}}_n \right) \right|$ is denoted by $\lambda^*$ and satisfies
    \begin{eqnarray*}
    \label{eq:optimal_lambda_gauss}
        \sum_{i=1}^{d} \dfrac{\lambda^* d - s_i}{ (s_i + \lambda^*) (s_i + 2 \lambda^*)} & = & 0\,, \quad  \text{ where } \, \,
%
       \frac{s_1}{d} \amp < \amp \lambda^* \amp < \frac{s_d}{d}.
       \end{eqnarray*}
\end{theorem}
\begin{proof}
    See the supplement.
\end{proof}
 Thus, for higher dimensional problems, the initial value of $\lambda$ can be chosen appropriately small. Theorem~\ref{thm:my_normal} also indicates that $\pi^{\lambda}$ is better conditioned than $\pi$, which in-turn improves the performance of the underlying Markov chains \citep{roberts1997updating}. This becomes particularly beneficial for HMC algorithms as we will see in our numerical studies.
\begin{remark}
    The optimal value of $\lambda$ in Theorem~\ref{thm:my_normal} is for iid samples from $\pi^{\lambda}$. Under this framework, in  the supplement 
    we also obtain the value of $n_e/n$ attained by this $\lambda^*$. Unsurprisingly, $n_e/n \to 1$ as $d \to \infty$. Such results are quite challenging to obtain for MCMC samples from $\pi^{\lambda}$, since a tractable expression of $\Sigma$ in $\Xi$ is unavailable.
\end{remark}
\section{Geometric ergodicity of Moreau-Yosida algorithms}
\label{sec:geom_erg}

We turn our attention to the conditions that guarantee MALA and HMC to yield a geometrically ergodic Markov chain for a Moreau-Yosida envelope density, $\pi^{\lambda}$. Properties of proximal operators help us arrive at simpler conditions for $\pi^{\lambda}$ than what is generally available for $\pi$. As is standard in the literature, we also highlight the convergence behavior for a class of one-dimensional distributions $\mathcal{E}(\beta, \gamma)$. This class of distributions was first discussed in \cite{roberts1996exponential} and provides a common ground to compare the performance with existing methods. For some $c \in \Real$ and constants $\gamma > 0$ and $0 < \beta < \infty$, $\pi \in \mathcal{E}(\beta, \gamma)$ is of the form,
\begin{eqnarray}  \label{one_dim_class}
    \pi(x) & \propto & \exp(-\gamma|x|^{\beta}) , \qquad |x| \geq c\,.
\end{eqnarray}
The value of $\beta$ controls how quickly the tails of the distribution decay. It is further assumed that $\pi$ is smooth enough for $|x| \leq c$ in order to satisfy basic differentiability assumptions.  We state a general result for distributions in this class.
\begin{result}[\cite{pereyra2016proximal}]
\label{res:one-d}
     Assume that $\pi \in \mathcal{E}(\beta, \gamma)$ with $\beta \geq 1$. Then $\pi^{\lambda} \in \mathcal{E} (\beta', \gamma')$ with $\beta' = \min(\beta, 2)$ and some $\gamma' > 0$.
\end{result}

\cite{pereyra2016proximal} does not provide the value of $\gamma'$, however we note that for distributions of the form \eqref{one_dim_class} having $\psi(x) = \gamma|x|^{\beta}$,
\begin{eqnarray} \label{eq:min_obj_fn}
    \psi^{\lambda}(x) & = & \min_{y \in \Real} \left\{\gamma|y|^{\beta} + \frac{1}{2\lambda}(x - y)^{2}\right\}\,.
\end{eqnarray}
When $\beta > 2$, the first term in the minimization \eqref{eq:min_obj_fn} dominates the second quadratic term pushing the minimization to occur at $y \approx 0$. When $\beta < 2$, the minimization tends to occur at $y \approx x$ due to the dominance of the quadratic term. The minimizer for $\beta = 2$ is obtained where the derivative of the objective function in \eqref{eq:min_obj_fn} vanishes. Consequently, for large $x$,
\begin{eqnarray}
   \gamma' & = & \gamma \mathbbm{1}(1 \leq \beta < 2) + \left(\frac{\gamma}{1+2\gamma\lambda}\right)\mathbbm{1}(\beta = 2) +  \left(\frac{1}{2\lambda}\right) \mathbbm{1}(\beta > 2)\,,
\end{eqnarray}
%
%
%
 where $\mathbbm{1}(.)$ denotes the indicator function. Result~\ref{res:one-d} ensures that when $\pi \in \mathcal{E}(\beta, \gamma)$, then $\pi^{\lambda} \in \mathcal{E}(\beta', \gamma')$. Critically, when $\pi$ has lighter than Gaussian tails, $\pi^{\lambda}$ has Gaussian tails. This will be critical to understanding the convergence behavior of $\pi^{\lambda}$-MALA and $\pi^{\lambda}$-HMC.
\subsection{Metropolis-adjusted Langevin algorithms}
\label{subsec:MALA_geom}

\cite{roberts1996exponential}  provide sufficient conditions for geometric ergodicity of a MALA algorithm given a general target $\pi$. We apply their results to $\pi^{\lambda}$. Let $q_\text{M}(x,y)$ denote the density of a MALA proposal for a target $\pi$, i.e.,  $q_\text{M}(x,y)$ is the density of $N(x + h \nabla \log \pi(x)/2, h \mathbbm{I}_d)$. Let $A(x)$ denote the acceptance region where a proposed value is guaranteed to be accepted:
\begin{eqnarray*}
     A(x) & = & \left\{y : \dfrac{\pi(y)q_{\text{M}}(y, x)}{\pi(x)q_{\text{M}}(x, y)} \geq 1\right\}\,.
\end{eqnarray*}
Let $R(x) = A(x)^{c}$ denote the potential rejection region, where there is a positive probability of rejection. Let $I(x)$ denote the set of points interior to $x$, i.e., $I(x) = \{y : \|y\| \leq \|x\|\}.$
%
%
Finally, let $A(x) \Delta I(x)$ denote the symmetric difference between $A(x)$ and $I(x)$.
\begin{theorem} [\cite{roberts1996exponential}]
\label{thm:mala_ge}
    Assume $A(\cdot)$ converges inwards in $q$, i.e.,
    \begin{eqnarray} 
    \label{eq:inwards_condn_mala}
     \lim_{\|x\| \rightarrow \infty} \int_{A(x) \Delta I(x)} q_{\text{M}}(x, y) dy & = & 0\,.
    \end{eqnarray}
Let $c(x) = x + h \nabla \log \pi(x)/2$ denote the mean of the proposal. If
    \begin{eqnarray} 
    \label{eq:lim_condn_ge}
        \eta & := & \liminf_{\|x\| \rightarrow \infty} \left\{\|x\| - \|c(x)\|\right\} \amp > \amp 0\,,
     \end{eqnarray}
then the MALA chain is geometrically ergodic for $\pi$.
\end{theorem}
Due to the specific structure of  $\nabla \log \pi^{\lambda}$, condition \eqref{eq:lim_condn_ge} can be met as long as the following assumption is satisfied.
\begin{assum}
\label{ass:limsup}
For a target density $\pi$ of the form \eqref{eq:targ_dens} with $\psi \in \Gamma(\Real^d)$, we assume
\begin{eqnarray}
\limsup_{\|x\| \to \infty} \dfrac{\|\prox^{\lambda}_{\psi}(x)\|}{\|x\|} \amp = \amp l \amp < \amp 1\,.
\end{eqnarray}
\end{assum}

\begin{theorem}
\label{thm:geom_mala}
     Let the target density be $\pi^{\lambda}$ so that $c(x) = x - h\nabla \psi^{\lambda}(x)/2$. Assume that $A(\cdot)$ converges inwards. Then if $h \leq 2\lambda$, the $\pi^{\lambda}$-MALA chain is geometrically ergodic.
\end{theorem}
\begin{proof}
    See the supplement.
\end{proof}
%
%



\begin{remark}
    \cite{pereyra2016proximal} does not employ the condition in  Assumption~\ref{ass:limsup} and instead state that for $\psi \in \Gamma(\Real^d)$, $\|\text{prox}^{\lambda}_{\psi}(x)\| < \|x\|$ for all $x$. Unfortunately, this is not true in general. Consider $\psi(x) = (x - 5)^2$, corresponding to a Gaussian density centered at $5$. Here
    \begin{eqnarray*}
        \text{prox}^{\lambda}_{\psi}(x) & = & x\dfrac{1}{1 + 2\lambda} + 5 \dfrac{2 \lambda}{1  + 2\lambda}\,.
    \end{eqnarray*}
    For $0 <x < 5$, $|\text{prox}^{\lambda}_{\psi}(x)| > |x|$, and this is true for many target densities that are not maximized at 0. Further, even when $\|\text{prox}^{\lambda}_{\psi}(x)\| < \|x\|$ is true for all $x$ (except $x \ne 0$), this is not sufficient to ensure $\eta$ in \eqref{eq:lim_condn_ge} is positive, since the $\liminf$ can still be zero.

    We instead employ Assumption~\ref{ass:limsup}. Many log-concave  densities satisfy this condition. A known exception is the Laplace density, for which even the original conditions of Theorem~\ref{thm:mala_ge} are not satisfied. However, models employing an $\ell_1$ prior with an $\ell_2$ likelihood satisfy this assumption. That is, for $a, b > 0$ and $x \in \Real$, consider $\psi(x) =  ax^2 + b|x|.$
Then
\begin{equation*}
\prox^{\lambda}_{\psi}(x) \amp = \amp \left(\frac{x - \lambda b}{1 + 2a\lambda}\right) \mathbbm{1}(x > \lambda b) + \left(\frac{x + \lambda b}{1 + 2a\lambda}\right) \mathbbm{1}(x < - \lambda b)\,,
\end{equation*}
which satisfies Assumption~\ref{ass:limsup} for all $\lambda > 0$.
Assumption~\ref{ass:limsup} is also sufficient to yield the geometric ergodicity results of \cite{pereyra2016proximal}.
\end{remark}
\begin{remark}
\label{rem:lis0}
    If $l = 0$ in Assumption~\ref{ass:limsup}, as it is often for light-tailed distributions, then it is possible to modify the proof so that under the assumption of $A(\cdot)$ converging inwards and with a choice of $h \leq 4 \lambda$, $\pi^{\lambda}$-MALA is geometrically ergodic.
\end{remark}
Theorem~\ref{thm:geom_mala} alleviates the burden of verifying \eqref{eq:lim_condn_ge} since a judicious choice of $h$ ensures condition \eqref{eq:lim_condn_ge} is satisfied.  Condition~\eqref{eq:inwards_condn_mala} is common and indicates that in the tails, the algorithm should either propose a value in the rejection region or propose a point moving away from the tails that is sure to be accepted. \cite{MR4003576} highlight that this condition is important to ``limit the degree of oscillation in the tails'' of the target density. \cite{roberts1996exponential} explain that the condition is connected to convexity of $\psi^{\lambda}$. Verification of the condition is typically done for different models on a case-by-case basis.
In addition to a general target, Theorem~\ref{thm:geom_mala} and Result~\ref{res:one-d}  highlight the improvements in mixing for $\pi^{\lambda}$ for when $\pi \in \mathcal{E}(\beta, \gamma)$.
\begin{enumerate}
    \item If $\beta \in [1, 2)$ then $\beta' = \beta$. For $h > 0$, MALA is geometrically ergodic for both $\pi$ and $\pi^{\lambda}$ by \citet[Section 16.1.3]{meyn2012markov} and  \citet[Theorem 4.1]{roberts1996exponential}.
    \item If $\beta = 2$ then $\beta' = 2$. According to \cite{roberts1996exponential}  $\pi^{\lambda}$-MALA is geometrically ergodic when $h\gamma' < 2$.  Therefore, $\pi^{\lambda}$-MALA is geometrically ergodic if $h < 4\lambda + 2/\gamma$.
    \item If $\beta > 2$ then by \eqref{eq:not_geom} $\pi$-MALA is not geometrically ergodic. However in this case, for $\pi^{\lambda}$, $\beta' = 2$ with $\gamma' = 1/(2 \lambda)$.  Consequently, the $\pi^{\lambda}$-MALA chain is geometrically ergodic as long as $h \leq 4\lambda$. This is consistent with Remark~\ref{rem:lis0}.
\end{enumerate}
For $\beta > 2$, the target $\pi$ corresponds to distributions with tails lighter than Gaussian. In such cases, efficient implementation of gradient-based schemes is difficult.  \cite{livingstone2022barker} proposed Barker's algorithm to mitigate this problem. In Section~\ref{sec:Poisson} we compare our MY-IS strategy with the Barker's algorithm as well. 
\subsection{Hamiltonian Monte Carlo} 
\label{sec:hmc}

Hamiltonian Monte Carlo is a gradient-based MCMC algorithm that has enjoyed success in a variety of problems. For a given target, $\pi(x) \propto \exp(-\psi(x))$, the HMC proposal is constructed by (approximately) conserving the total energy or Hamiltonian,
\begin{eqnarray*}
    H(x, z) & = & \psi(x) + \frac{1}{2}z^{\text{T}}M^{-1}z\,.
\end{eqnarray*}
Here, $z$ is a $d$-dimensional augmented momentum variable \citep[see][for details]{MR2858447} and $M$ is a $d \times d$ positive-definite mass matrix. We will fix $M=\mathbbm{I}_{d}$. Let $L \geq 1$ be an integer representing the number of leapfrog steps in the approximation of the Hamiltonian dynamics and let $\varepsilon > 0$ be a step-size. For a current state $x_0$ and $z_0 \sim N(0, \mathbbm{I}_d)$, \cite{MR4003576} state that the HMC proposal for a target $\pi$ can be expressed as
    \begin{eqnarray}
        x_{L\varepsilon} & = & x_{0} - \frac{L\varepsilon^{2}}{2}\nabla \psi(x_{0}) - \varepsilon^{2}\sum_{i=1}^{L-1} (L-i)\nabla \psi(x_{i\varepsilon}) + L\varepsilon z_{0} \,,
    \end{eqnarray}
where $x_{i\varepsilon}$ is the state at the $i^{\text{th}}$ leapfrog step. We will denote the mean of the proposed value as $m_{L, \varepsilon}(x_{0}, z_{0}) = x_{0} - (L\varepsilon^{2}/2)\nabla \psi(x_{0}) - \varepsilon^{2}\sum_{i=1}^{L-1} (L-i)\nabla \psi(x_{i\varepsilon})$, so that $x_{L \epsilon} = m_{L, \varepsilon}(x_{0}, z_{0}) + L \epsilon z_0$. Let $q_{\text{H}}(x_0, \cdot)$ denote the density of the HMC proposal. Define the potential rejection region as
\begin{eqnarray}
    R(x) & = &  \left\{y : \dfrac{\pi(y)q_\text{H}(y, x)}{\pi(x)q_\text{H}(x, y)} \leq  1  \right\}\,.
\end{eqnarray}
Sufficient conditions for geometric ergodicity for HMC were provided in \cite{MR4003576}. The following assumption  on $L$ ensures irreducibility of the Markov chain.

\begin{assum}
\label{ass:L}
The leapfrog steps $L$ is chosen using a distribution $\mathcal{L}(\cdot)$ such that $\Pr_{\mathcal{L}}[L=1] > 0$ and for any $(x_{0}, z_{0}) \in \mathbb{R}^{2d}$ and $\varepsilon > 0$, there is an $s < \infty$ such that $\mathbb{E}_{\mathcal{L}}[e^{s\|x_{L\varepsilon}\|}] < \infty$.
\end{assum}

\begin{theorem} [\cite{MR4003576}]
\label{thm:living_hmc}
    The HMC algorithm produces a $\pi$-geometrically ergodic Markov chain if Assumption~\ref{ass:L} holds and for $\eta(d) = \Gamma((d+1)/2)/\Gamma(d/2)$, $1/2 < \delta < 1$,
    \begin{eqnarray} 
    \label{limsup_condn}
        \limsup_{\|x\| \rightarrow \infty, \|z\| \leq \|x\|^{\delta}} \left\{\|m_{L, \varepsilon}(x, z)\| - \|x\|\right\} & < & -\sqrt{2}L\varepsilon \eta(d)\, \qquad \text{    and} \\ 
        \label{eq:inwards_condn_hmc}
         \lim_{\|x\| \rightarrow \infty} \int_{R(x) \cap I(x)} q_{\text{H}}(x, y) dy & = & 0 \,.
    \end{eqnarray}
    %
    %
    %
\end{theorem}

Condition \eqref{limsup_condn} can be  challenging to verify, but properties of $\pi^{\lambda}$ make this task easier.
%
%
\begin{theorem}
\label{thm:geom_hmc}
    Suppose Assumption~\ref{ass:limsup} and Assumption~\ref{ass:L} hold. Then  $\pi^{\lambda}$-HMC algorithm is geometrically ergodic for a sufficiently small $\varepsilon$, if \eqref{eq:inwards_condn_hmc} holds for $\pi^{\lambda}$.
    %
    %
\end{theorem}
\begin{proof}
    See the supplement.
\end{proof}
We recognize that as it is true for general targets, demonstrating \eqref{eq:inwards_condn_hmc} can be challenging to establish for general problems. We have, however, some guiding principles based on $\pi$'s tail behavior when $\pi$ belongs to the exponential family class $\mathcal{E}(\beta, \gamma)$.
   \begin{enumerate}
    \item If $\beta \in [1, 2)$ then $\beta' = \beta$. By the result of \cite{MR4003576}, $\pi^{\lambda}$-HMC is geometrically ergodic.
    \item If $\beta = 2$ then $\beta' = 2$ and $\pi$-HMC is geometrically ergodic for sufficiently small  $\varepsilon$. Consequently, $\pi^{\lambda}$-HMC is also geometrically ergodic for sufficiently small $\varepsilon$.
    \item If $\beta > 2$ then $\pi$-HMC is not geometrically ergodic \citep{MR4003576}. However for $\pi^{\lambda}$, $\beta' = 2$, and thus $\pi^{\lambda}$-HMC is geometrically ergodic for sufficiently small $\varepsilon$.
\end{enumerate}
\section{Numerical studies}
\label{sec:numerical}

We evaluate our proposed methodology in a variety of simulation studies. 
%
In all studies we tune $\lambda$ following the guidelines in Section~\ref{sec:choosing_lambda}. We implement MY-IS with both $\pi^{\lambda}$-MALA and $\pi^{\lambda}$-HMC samples. We compare Markov chains $\pi^{\lambda}$-MALA with P-MALA of \cite{pereyra2016proximal} and $\pi^{\lambda}$-HMC with P-HMC of \cite{chaari2016hamiltonian}.  For the Bayesian Poisson random effects model, \cite{livingstone2022barker} show that the traditional MALA algorithm is unreliable and demonstrate the superiority of their Barker's algorithm. We compare $\pi^{\lambda}$-MALA and $\pi^{\lambda}$-HMC chains  with the Barker's algorithm as well. We employ a warm start for all Markov chains; other implementation details of these algorithms are provided in the supplement.

We evaluate methods on their statistical efficiency
by comparing variances of estimators of the posterior means. The estimated relative efficiency of Method 1 over Method 2 is
\begin{eqnarray*}
    \text{eff}^{\text{rel}} & = &  \dfrac{1}{p}\sum_{i=1}^{p} \dfrac{\hat{\tau}^{2}_{i({\text{Method 2}})}}{\hat{\tau}^{2}_{i({\text{Method 1}})}}\,,
\end{eqnarray*}
where $\hat{\tau}^{2}_{i(\cdot)}$ denotes the estimated asymptotic variance of component $i$ from the respective method. Variance of MY-IS estimators are estimated using the methods described in the supplement. R code is available at \url{https://github.com/sapratim/MoreauYosidaMCMC-IS}.

\subsection{Toy example}
\label{sec:toy_exam}

Consider the target distribution with density
\begin{eqnarray}
    \label{eq:laplace}
    \pi_{\beta}(x_1, \dots, x_d) & \propto & \prod_{i=1}^{d} e^{-\psi_{\beta}(x_i)} \amp = \amp \prod_{i=1}^{d} e^{-|x_i|^{\beta}}\,, \qquad x_{i} \in \Real\,,
\end{eqnarray}
$i = 1, 2, \ldots, d$, for $\beta = 1$ (Laplace) and $\beta = 4$ (super-Gaussian). Details of the Moreau-Yosida envelopes are provided in the supplement. 

We investigate the effect of $\lambda$ on the $\pi^{\lambda}$-MALA Markov chain and the quality of the importance sampling estimator $\hat{\theta}^{\text{MY}}_n$ for $d = 20$ (the supplement contains results for $d = 1, 10, 20$). For every combination of $\beta$ and $d$, we track the following as $\lambda$ is varied: (i) the estimated $n_e/n$, which reflects the quality of the importance sampling proposal, $\pi^{\lambda}$, (ii) the MCMC effective sample size by $n$ for estimating the mean of $\pi^{\lambda}$, which reflects the mixing quality of the $\pi^{\lambda}$-Markov chain, and (iii) the estimated asymptotic variance of $\hat{\theta}_n^{\text{MY}}$. Each Markov chain is run for $n = 10^6$ iterations. Figure~\ref{fig:toy} shows these quantities as a function of $\lambda$ for both $\beta$.

Similar trends hold for both values of $\beta$. As $\lambda$ increases, $n_e/n$ decreases, the MCMC effective sample size for $\pi^{\lambda}$ improves and the trade-off between these two is balanced by the variance of $\hat{\theta}_n^{\text{MY}}$. These results motivate our recommendation to choose $\lambda$ such that $n_e/n$ is in $[0.40, 0.80]$. 
\begin{figure} [h]
    \includegraphics[scale = .52]{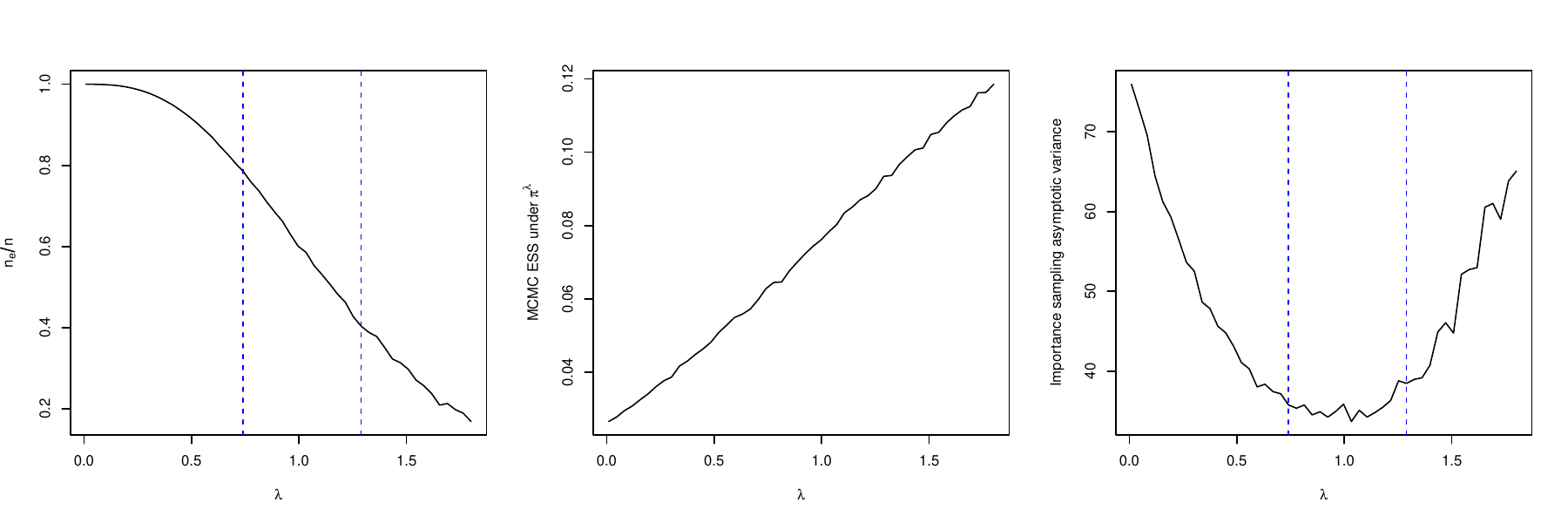}
     \includegraphics[scale = .52]{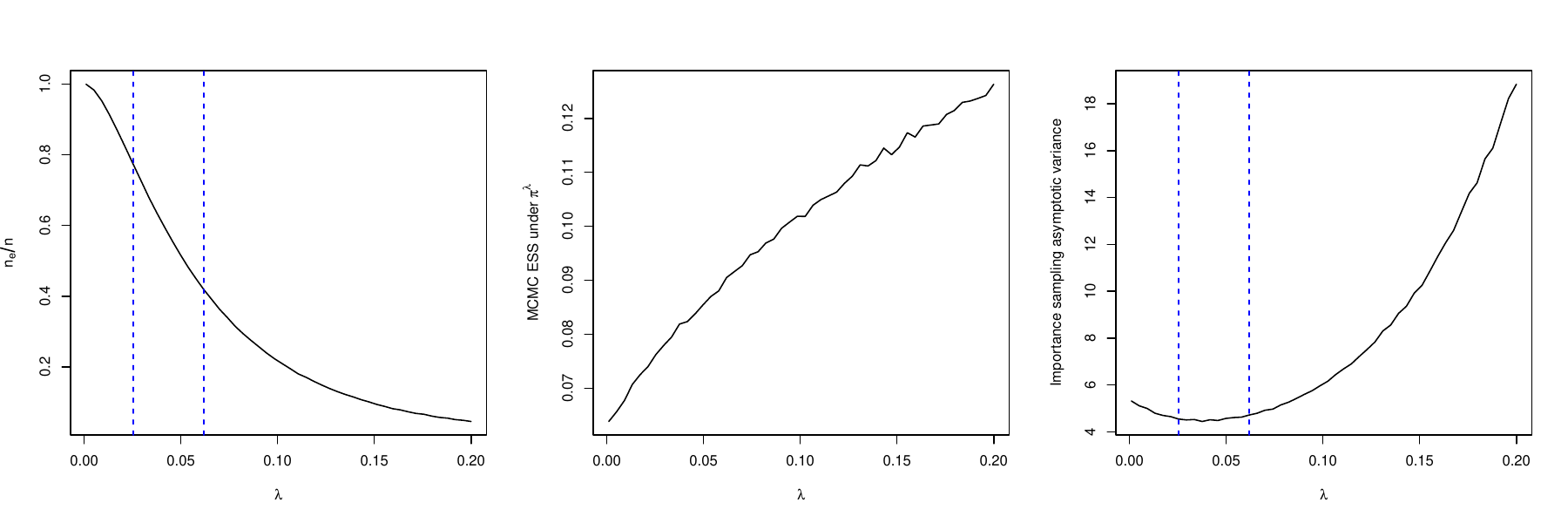}
    \caption{Top (Laplace) and bottom (Super-Gaussian) for $d = 20$. Left column plots $n_e/n$ for different $\lambda$, middle column plots the MCMC effective sample size for the $\pi^{\lambda}$-MALA chain for different $\lambda$, and the right column plots the estimated importance sampling asymptotic variance for different values of $\lambda$. The two vertical lines are the values of $\lambda$ that yield $n_e/n \in \{0.40, 0.80\}$.}
    \label{fig:toy}
\end{figure}
%
%
%
%
\subsection{Bayesian trendfiltering}
\label{sec:trendfiltering}
Consider the standard task in nonparametric regression. 
Let $y(t)$ be a scalar function of time that is a superposition of a smooth function $\mu(t)$ and additive noise. Suppose $y(t)$ is observed at  time points $t_1, \ldots, t_m$. Then $y =  \mu + \epsilon$,  where $y = (y(t_1), \dots, y(t_m))^{\Tra}$
%
%
, $\mu = (\mu(t_1), \dots, \mu(t_m))^{\Tra}$ and $\epsilon \sim N(0, \sigma^{2}\mathbb{I}_{m})$. The goal is to recover $\mu \in \Real^{m}$ from the observations $y \in \Real^m$. 

Let $\text{D}_{m}^{(k+1)} \in \Real^{(m - (k+1)) \times m}$ be a discrete difference matrix of order $(k+1)$ and dimension $m$. \cite{kim2009ell_1} proposed  $\ell_{1}$-trendfiltering, which estimates $\mu$ with the solution to the following convex optimization problem
\begin{equation} 
\label{eq:tf_eqn}
    \underset{\mu \in \Real^{m}}{\text{minimize}}\;  \frac{1}{2}\|y - \mu\|^{2} + \alpha\left\|\text{D}_{m}^{( k+1)}\mu\right\|_{1}\,.
\end{equation}
%
%
%
The function $\left\lVert \text{D}_{m}^{(1)} \mu \right\rVert_1 = \sum_{i = 1}^{m-1} |\mu_{i+1} - \mu_{i}|$ is the fused lasso penalty \citep{tibshirani2005sparsity}, which is designed to recover piecewise constant $\mu$. When $k = 1, 2,$ and $3$, the penalty incentivizes the recovery of piecewise linear, quadratic, and cubic trends, respectively. Difference matrices  can be calculated recursively by the relation $\text{D}_{m}^{(k+1)} = \text{D}_{m-k}^{(1)} \text{D}_{m}^{(k)}$. 
The solution to \eqref{eq:tf_eqn} produces the maximum a posteriori (MAP) estimator  corresponding to an appropriate Bayesian model. Later works  proposed approaches that not only produce point estimates but also can quantify the uncertainty in those estimates \citep{roualdes2015bayesian,faulkner2018locally,kowal2019dynamic,heng2023bayesian}. 
We consider a Bayesian model with the following posterior distribution that corresponds to the MAP estimate computed in \eqref{eq:tf_eqn},
\begin{eqnarray}
\label{eq:tf_posterior}
    \pi(\mu | y) & \propto & \exp \left\{ - \dfrac{\|y - \mu\|^{2}}{2\sigma^2}- \alpha\left\|\text{D}_{m}^{(k+1)} \mu\right\|_{1}  \right\}\,.
\end{eqnarray}
The $\ell_1$-penalty renders the posterior non-differentiable, precluding the use of traditional gradient-based MCMC schemes. Thus, the posterior in \eqref{eq:tf_posterior} is a natural candidate for a proximal MCMC-type sampler. We set $\sigma^2 = 9$ and for $m = 100$ obtain equally spaced time points from
\begin{eqnarray*}
   \mu(t) \amp = \amp t \mathbbm{1}(t \leq 35) + (70 - t) \mathbbm{1}(t \leq 70) + (0.5t - 35) \mathbbm{1}(t > 70)\,.
\end{eqnarray*}
%
%
%
Figure~\ref{tf_visualisation} shows the scatter plot of observed values, posterior mean and a band of $95\%$ credible intervals computed by MY-IS, with the latter derived using  \eqref{eq:impsamp_quantile}.
\begin{figure}
    \centering
    \includegraphics[scale = 0.35]{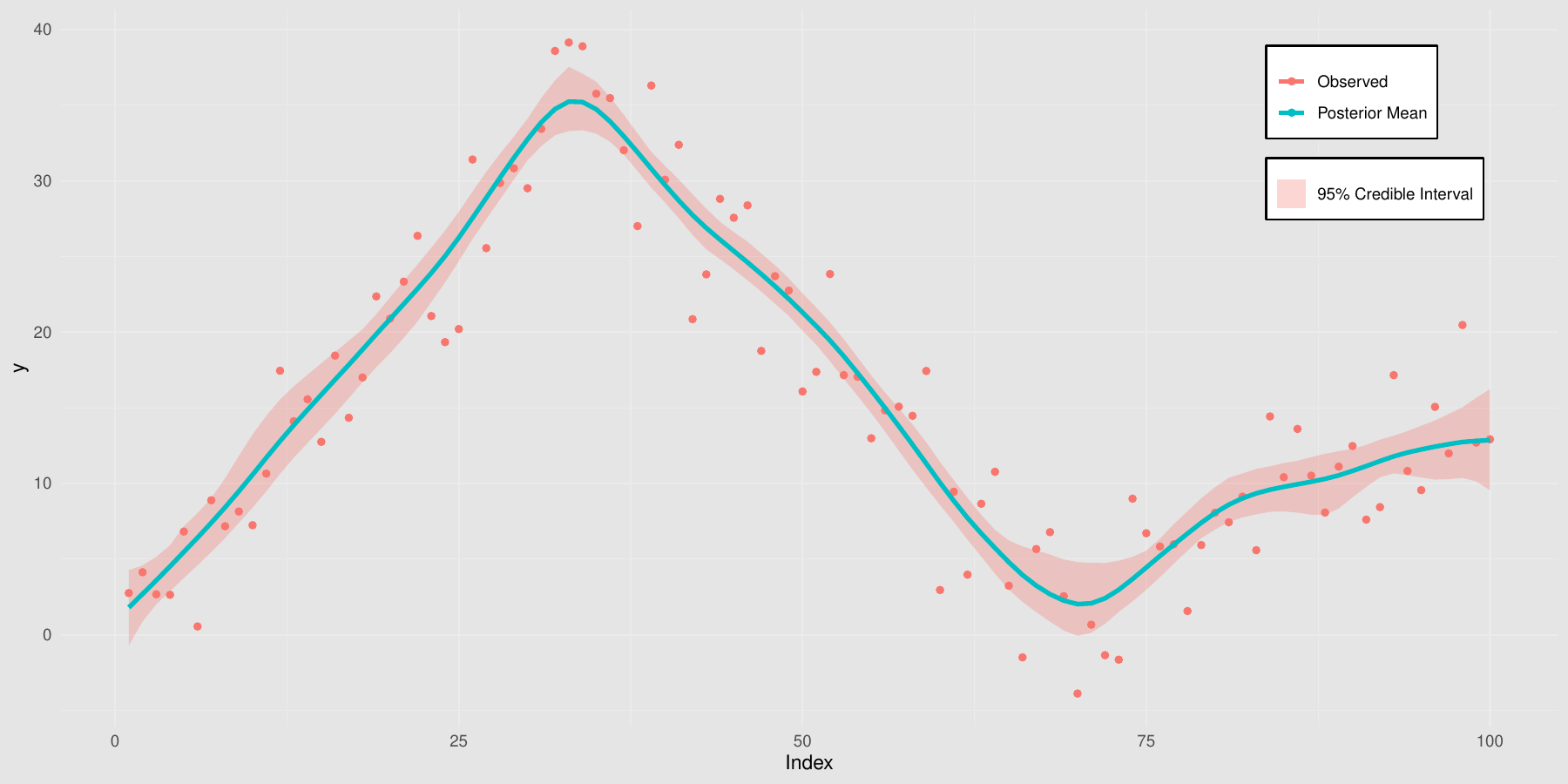}
    \caption{Trend filtering fit using the marginal quantiles and posterior mean for importance sampling estimator.}
    \label{tf_visualisation}
\end{figure}

We compare four proximal MCMC algorithms: (i) P-MALA, (ii) P-HMC, (iii) $\pi^{\lambda}$-MALA, and (iv) $\pi^{\lambda}$-HMC. For the latter two chains,  we set $\lambda  = 0.001$ to  obtain an importance sampling effective sample size of $n_e/n \approx 0.47$. For all chains we simulate a Monte Carlo sample size of $n = 10^5$. We ran 100 replications of all four chains to ascertain the gains in relative efficiencies. 

We denote the MY-IS estimator constructed from the $\pi^{\lambda}$-MALA and $\pi^{\lambda}$--HMC chains as MYIS-MALA and MYIS-HMC, respectively. Figure~\ref{fig:tf_rel_eff} shows the average relative efficiency of P-MALA versus MYIS-MALA (left) and P-HMC versus MYIS-HMC (right) in box plots of the 100 components. For HMC, estimation of components are at least 25 times more efficient. For MALA it is at least 2.6 times more efficient. The gains in efficiency with importance sampling are significant, especially in the case of $\pi^{\lambda}$-HMC.
\begin{figure}
    \centering
    \includegraphics[scale = 0.48]{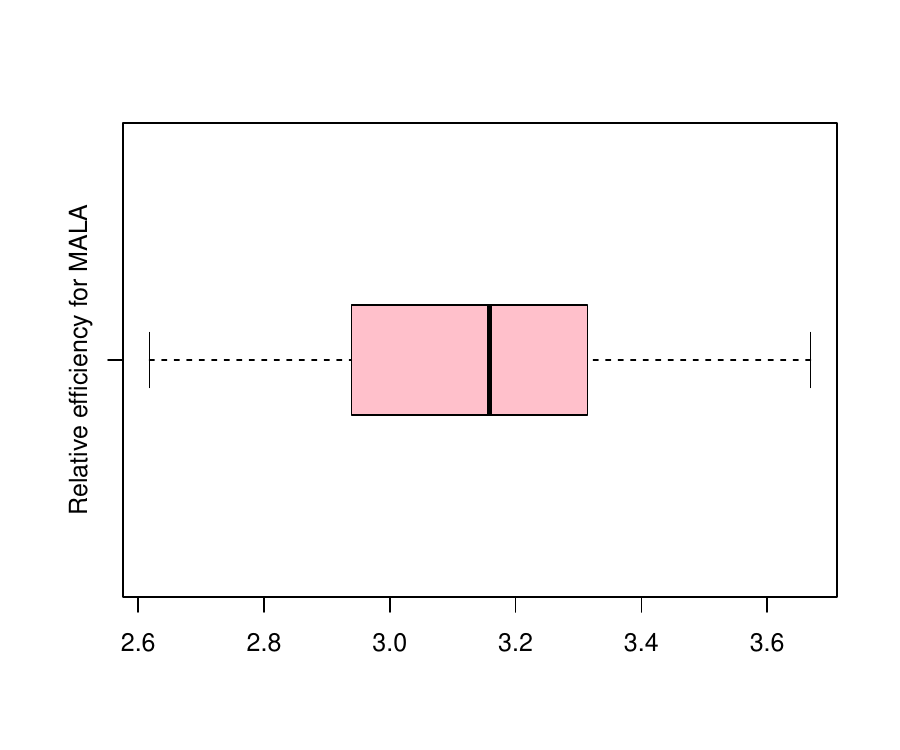}
    \includegraphics[scale = 0.48]{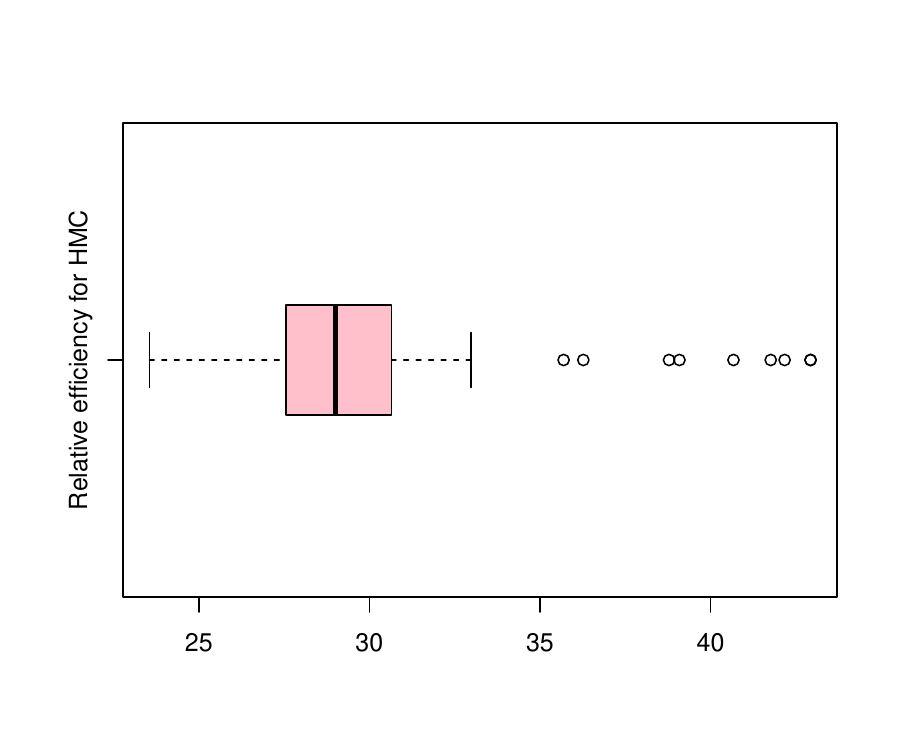}
    \caption{Trendfiltering: Average relative efficiencies of MYIS-MALA over P-MALA (left) and MYIS-HMC over P-HMC (right).}
    \label{fig:tf_rel_eff}
\end{figure}
To better understand these differences in efficiency gains, we examine the ability of MALA and HMC algorithms to explore $\pi^{\lambda}$ compared to $\pi$. Specifically, we compute the difference in the estimated autocorrelation lags between pairs of algorithms, i.e., the difference between their autocorrelation functions (ACFs). Figure~\ref{fig:acf_tf} presents the ACF difference plot computed using samples from P-MALA and $\pi^{\lambda}$-MALA (left) and P-HMC and $\pi^{\lambda}$-HMC (right) in a box plot over the 100 components of the chains.  For MALA, some components show improved mixing, whereas other components show marginally slower mixing. Despite this, the use of importance sampling yields significant gains as was seen in Figure~\ref{fig:tf_rel_eff}. In contrast, for HMC, all components exhibit substantially improved mixing of $\pi^{\lambda}$-HMC.  
\begin{figure}[t] 
    \centering
    \includegraphics[scale = 0.48]{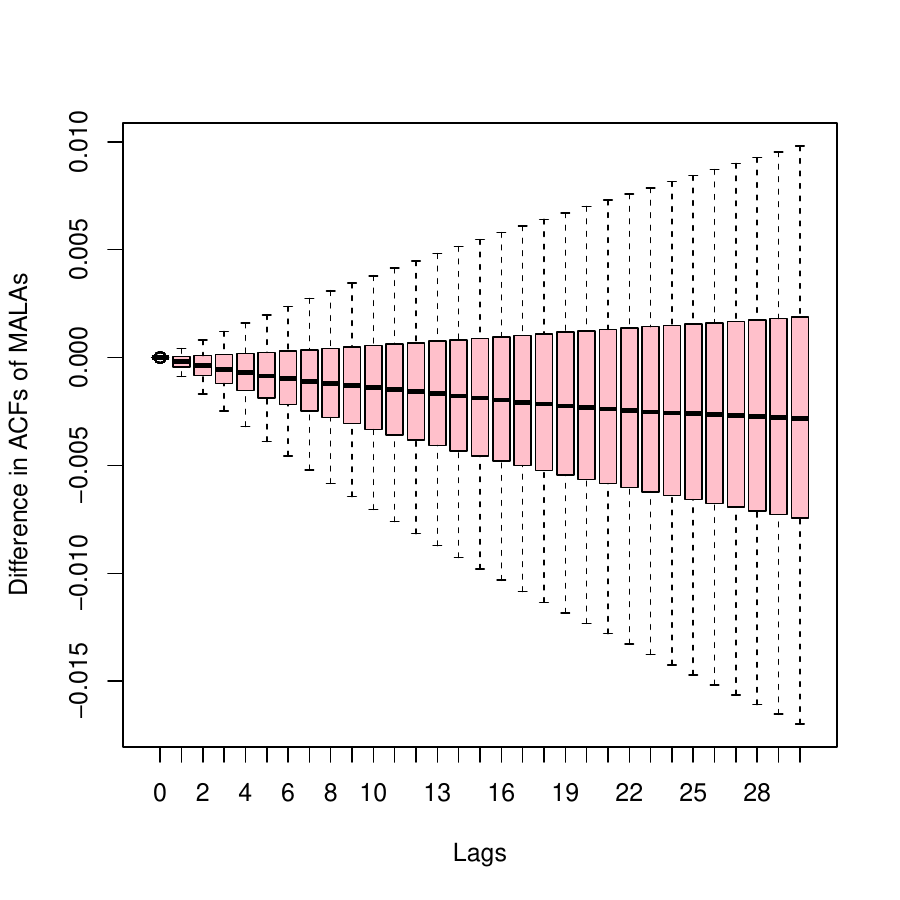}
    \includegraphics[scale = 0.48]{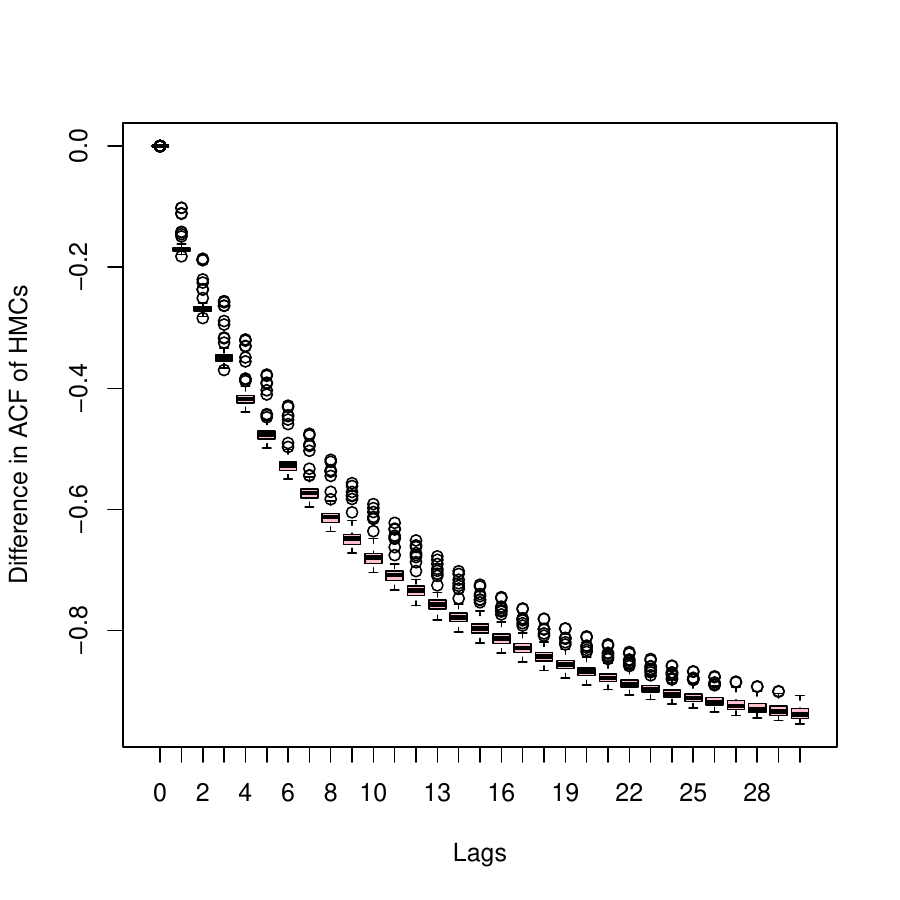}
    \caption{Trendfiltering: ACF difference plots computed using samples from  $\pi^{\lambda}$-MALA and P-MALA (left), $\pi^{\lambda}$-HMC and P-HMC (right).}
    \label{fig:acf_tf}
\end{figure}
\subsection{Nuclear-norm based low rank matrix estimation}
\label{sec:nuc_norm}

Estimating low rank matrices is a fundamental problem in statistics. A classic penalized likelihood approach to recovering low rank matrices is to employ a nuclear norm regularization term \citep{fazel2002matrix}. One of the most successful uses of the nuclear norm has been used in this role is in matrix completion \citep{Mazumder2010, Cai2010}. For illustrative purposes, we consider the problem of estimating a low-rank matrix in the context of matrix denoising.

We observe a matrix $Y \in \Real^{m \times k}$ where $Y = X + E$. The matrix $X$ is a latent low-rank matrix that we wish to recover and $E$ is a noise matrix with iid Gaussian entries, i.e., $e_{ij} \sim N(0, \sigma^2)$. We assume the prior $\pi(X)   \propto  \exp(-\alpha\|X\|_{*})$ on $X$
%
%
where $\|X\|_{*}$ is the nuclear norm of $X$.
Having observed $Y$, the posterior density of $X$ is 
\begin{eqnarray} 
\label{eq:matrix_post}
    \pi(X|Y) & \propto & \exp\left\{-\left(\frac{\|Y - X\|_{\text{F}}^{2}}{2\sigma^{2}} + \alpha\|X\|_{*} \right) \right\}\,,
\end{eqnarray}
where $\|X\|_{\text{F}}$ denotes the Frobenius norm of $X$. 
%
As in the trendfiltering example, a non-smooth penalty renders the target posterior non-differentiable, precluding the use of traditional gradient-based MCMC schemes. Thus, the posterior in \eqref{eq:matrix_post} is also a natural candidate for a proximal MCMC-type sampler.
\cite{recht2010guaranteed} derived the proximal mapping of the  negative log-posterior in \eqref{eq:matrix_post}:
\begin{equation}
    \prox_{\psi}^{\lambda}(X) \amp = \amp \text{SVT} \left( \dfrac{\lambda}{\lambda +  \sigma^{2}}Y + \dfrac{\sigma^2}{\lambda + \sigma^2}X, \dfrac{\alpha\lambda\sigma^{2}}{\lambda + \sigma^{2}}  \right)\,.
\end{equation}
The mapping $\text{SVT}(Z, t)$ is the singular value soft thresholding operator. Let $Z = U D V\Tra$ denote a singular value decomposition of $Z$ where $D$ is a diagonal matrix of singular values. Let $d_i$ denote the $i^{\text{th}}$ singular value of $D$. Let $\tilde{D}$ be the matrix obtained by replacing the $i^{\text{th}}$ singular value of $D$ by 
$\max\{d_{i} - t, 0\}$. Then $\text{SVT}(Z, t) = U\tilde{D}V\Tra$.
In the following numerical studies, we take $X$ to be the checkerboard image of $64 \times 64$ pixels studied in \cite{pereyra2016proximal}. Note that this is a relatively high-dimensional Bayesian inference problem as the dimension of the posterior distribution is 4096 ($64^{2}$). Additional details on the data are in the supplement.
We set $\sigma^2 = 0.01$ and following \cite{pereyra2016proximal}, set  $\alpha = 1.15/\sigma^{2}$. Setting $\lambda = 10^{-4}$ we obtain an importance sampling effective sample size of $n_e/n \approx 0.41$. For all chains we simulate a Monte Carlo sample size of $n = 10^5$ and ran 100 replications of all chains to ascertain the gains in relative efficiencies.
We present ACF difference plots from a randomly selected single replicate. For the ACF plots we obtain component-wise autocorrelations and present boxplots of component-wise difference in ACFs of $\pi^{\lambda}\text{-MALA} - \text{P-MALA}$ (same for HMC). For the relative efficiencies, we average the relative efficiencies over the 100 replications and present a boxplot of the average relative efficiencies over all components.
\begin{figure}
    \centering
    \includegraphics[scale = 0.43]{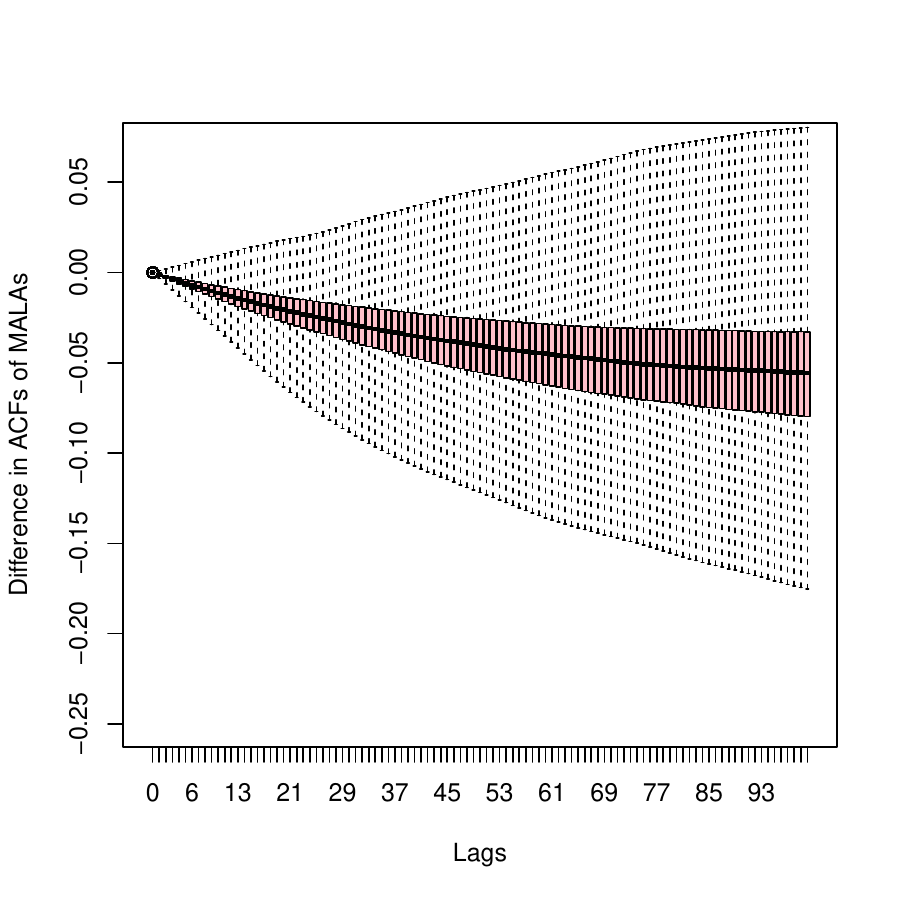}
    \includegraphics[scale = 0.43]{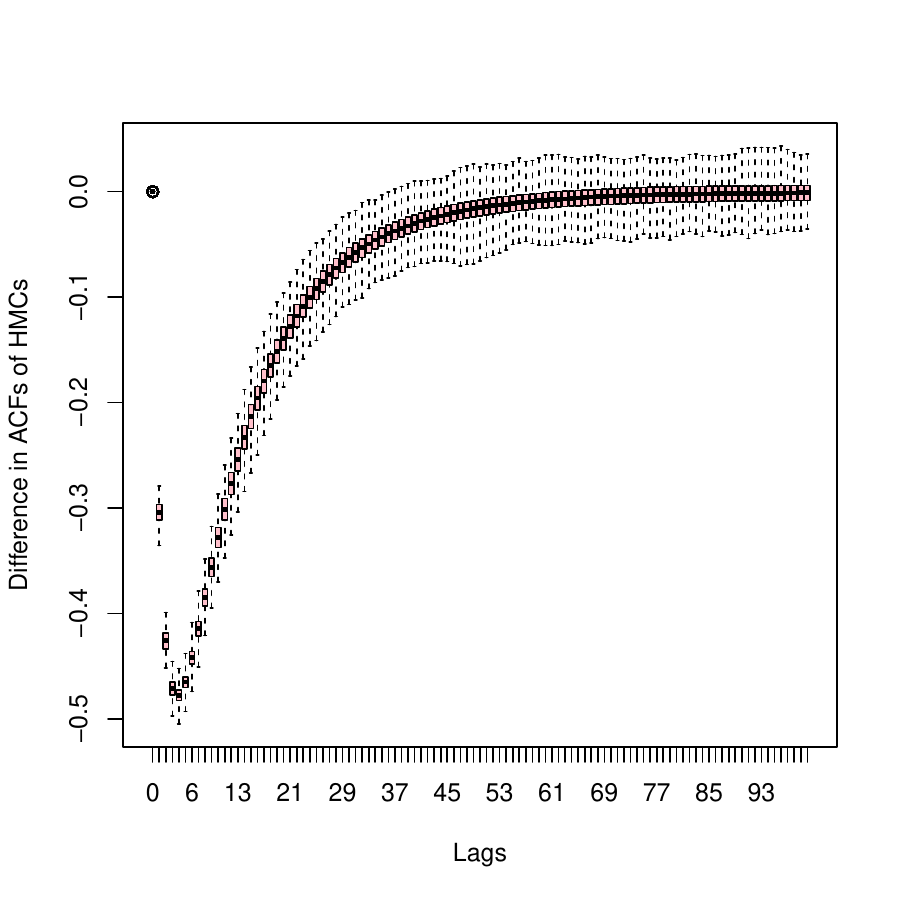}
    \caption{Nuclear-norm based matrix denoising: ACF difference plots computed using  samples from $\pi^{\lambda}$-MALA and P-MALA (left), $\pi^{\lambda}$-HMC and P-HMC (right).}
    \label{fig:acf_nn}
\end{figure} 

As previously observed in \cite{pereyra2016proximal}, P-MALA works well in this example. This can be seen in Figure~\ref{fig:acf_nn}, where the difference in the ACFs for the P-MALA and $\pi^{\lambda}$-MALA chains is marginally significant. However, Figure~\ref{fig:boxplot_nn} indicates that the MY-IS procedure nonetheless yields a more efficient estimator. For HMC, the efficiency gains are even better. Figure~\ref{fig:acf_nn} indicates a significant improvement in the ACF behavior for the $\pi^{\lambda}$-HMC chain over the P-HMC chain, which leads to a significant gain in relative efficiency as demonstrated in Figure~\ref{fig:boxplot_nn}.

\begin{figure}[t]
    \centering
    \includegraphics[scale = 0.43]{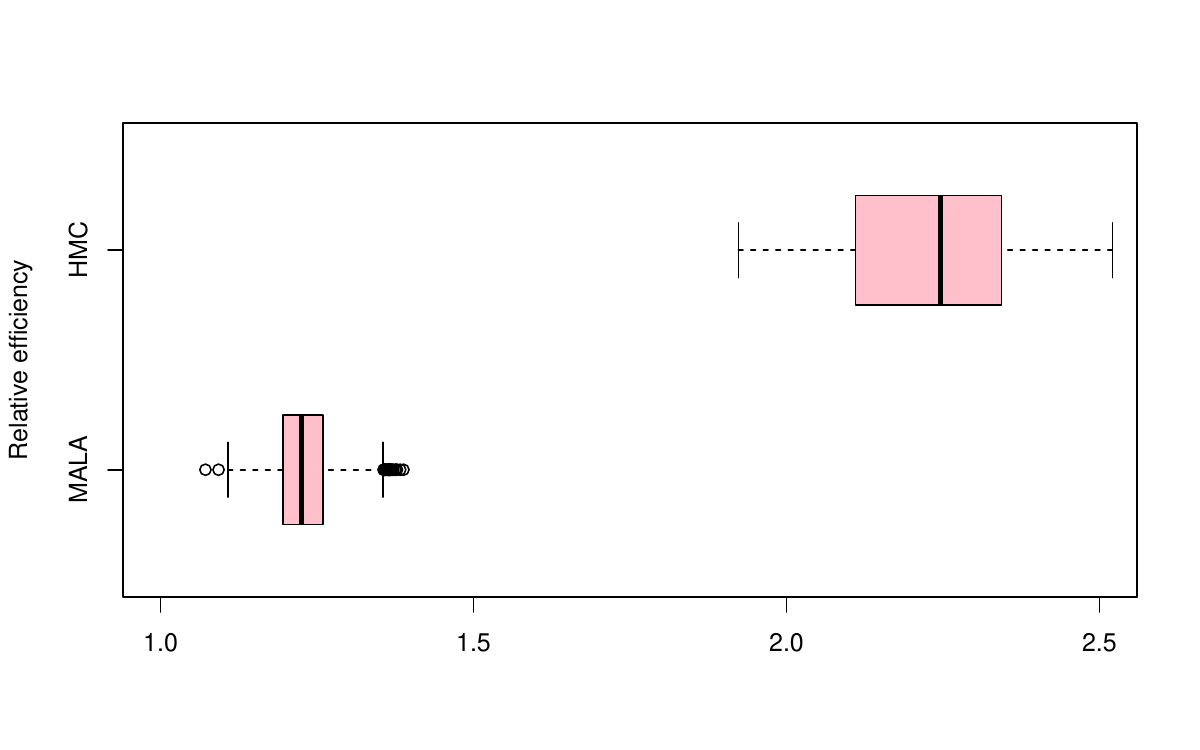}
    \caption{Nuclear-norm based matrix denoising: Average relative efficiencies of MYIS-MALA over P-MALA (left) and MYIS-HMC over P-HMC (right).}
    \label{fig:boxplot_nn}
\end{figure}
\subsection{Bayesian Poisson random effects model}
\label{sec:Poisson}

We revisit the Bayesian Poisson random effects model presented in \cite{livingstone2022barker}. MALA and HMC struggle to reliably generate samples from this model's posterior due to its light tails. Their proposed Barker's algorithm, however, can successfully generate samples in spite of the light tails. Given the effect of Moreau-Yosida smoothing on the exponential family class $\mathcal{E}(\beta, \gamma)$, one might conjecture that our MY-IS scheme could also be robust to light tails. Consequently, we test this conjecture by estimating the posterior mean with our MY-IS scheme. 

The model has the following hierarchical specification:
\begin{align*}
    y_{ij} \mid \eta_{i} \amp  &\sim \amp \text{Poisson}(e^{\eta_{i}}) 
 \qquad j = 1,2, \ldots, n_{i},\\
  \mu \amp  \sim \amp  \text{N}(0, c^{2}) \quad &\text{ and } \quad   \eta_{i} \mid \mu \amp  \sim \amp  \text{N}(\mu, \sigma_{\eta}^{2})  \qquad \qquad i = 1,2, \ldots, I, 
\end{align*}
where $y_{ij}$ represent count data measured for the $j^{\text{th}}$ subject in the $i^{\text{th}}$ class, with $n_i$ being the number of subjects in the the $i^{\text{th}}$ class. Following \cite{livingstone2022barker}, we set the number of classes $I$ to be 50, $\sigma_{\eta} = 3$, and $c = 10$. In the supplement,
we present an efficient algorithm to evaluate the Moreau-Yosida envelope. In contrast to the previous examples, this posterior density is differentiable. Thus, we also compare our results with those  obtained using Barker's algorithm. We set $\lambda = 0.001$ to obtain  an importance sampling effective sample size of  $n_e/n \approx 0.41$. For all chains we simulate a Monte Carlo sample size of $n = 10^5$ and ran 100 replications of all chains to ascertain the gains in relative efficiencies.
\begin{figure}
    \centering
    \includegraphics[scale = 0.32]{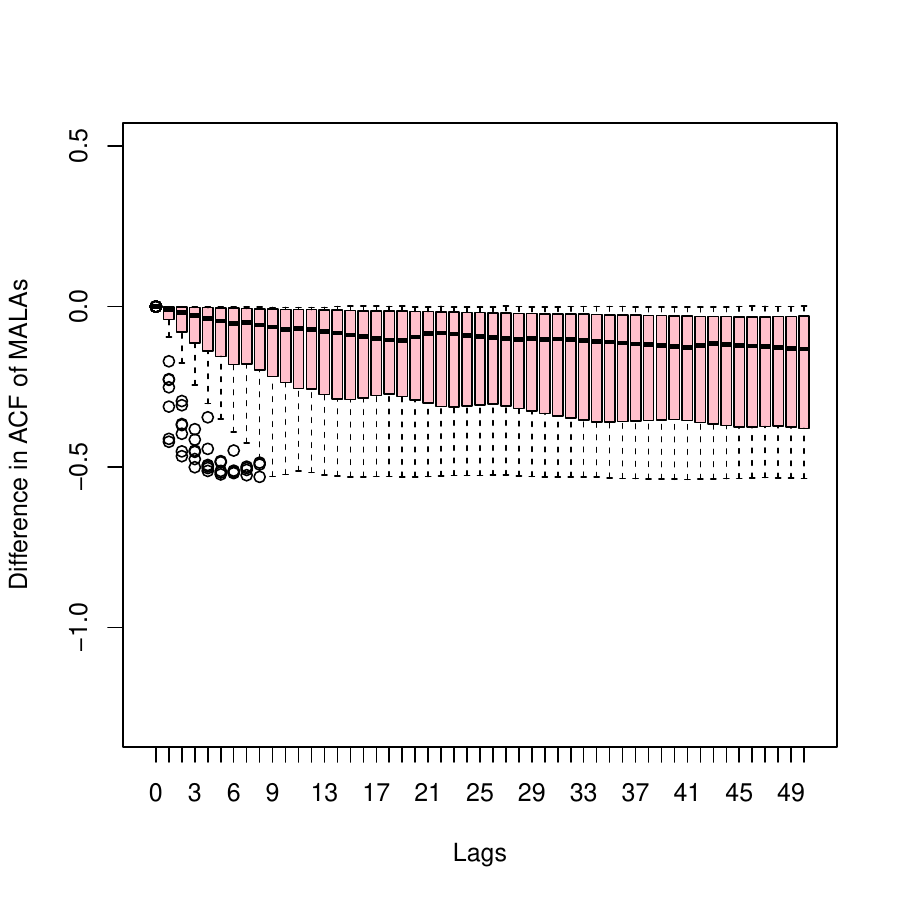}
    \includegraphics[scale = 0.32]{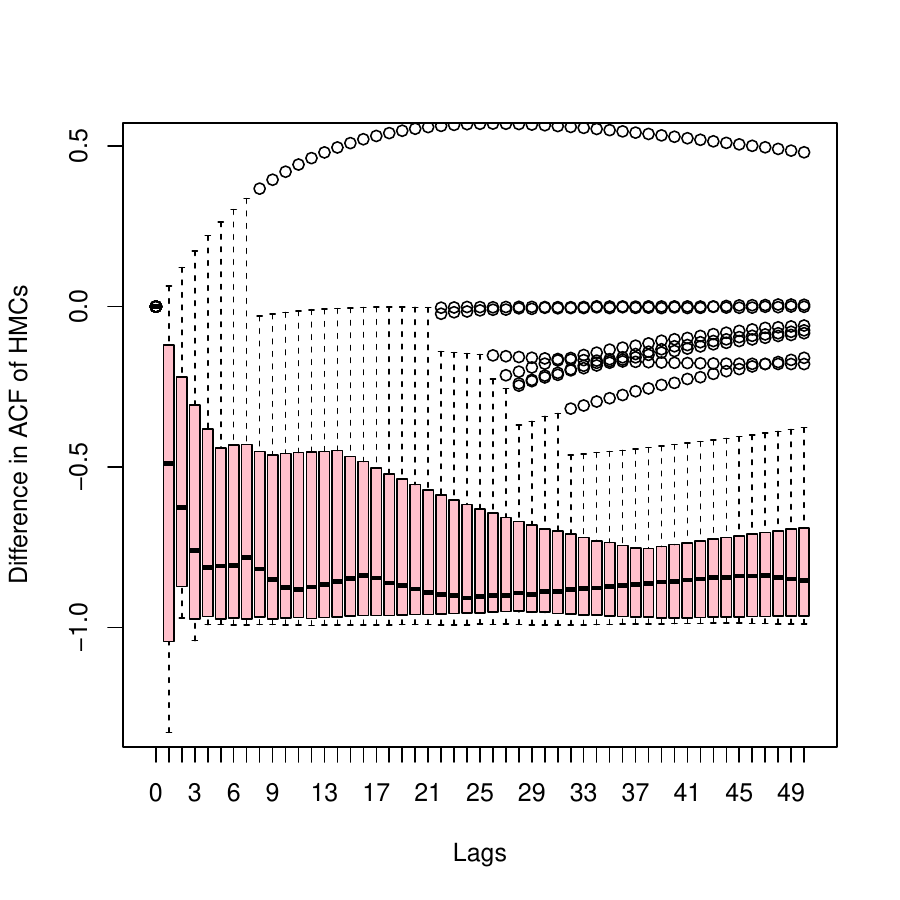}
    \includegraphics[scale = 0.32]{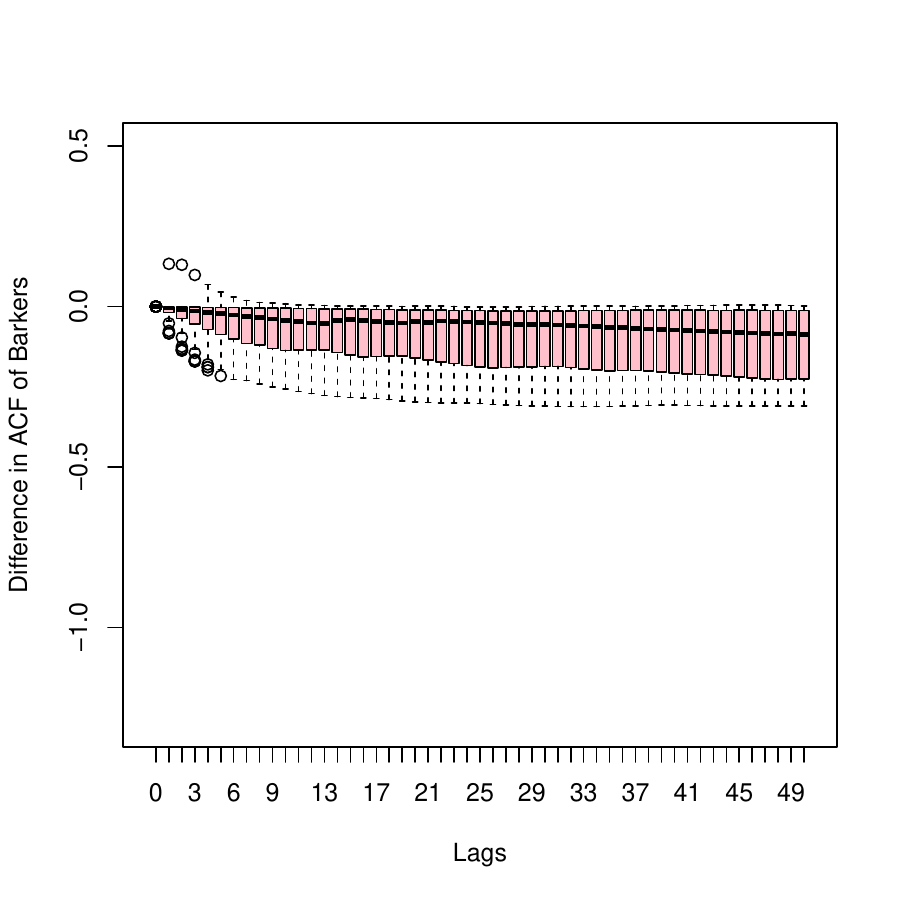}
    \caption{Poisson model: ACF difference plots computed using  samples from $\pi^{\lambda}$-MALA and P-MALA (left), $\pi^{\lambda}$-HMC and P-HMC (center) and $\pi^{\lambda}$-Barker and P-Barker (right).}
    \label{fig:acf_poisson}
\end{figure}
As before, we show the ACF difference plots from a randomly selected single replicate. This time we also compare Barker's algorithm with $\pi^{\lambda}$-Barker's algorithm.  Figure~\ref{fig:acf_poisson} shows that all three $\pi^{\lambda}$ chains mix better than their counterparts in general. There is one component whose mixing is better in P-HMC and Barker's, but for all other components there is a significant improvement in the quality of the Markov chains. Particularly, $\pi^{\lambda}$-HMC demonstrates significant improvement.

To demonstrate the impact of using MY-IS, we also present the relative efficiencies, focusing also on the relative efficiency of MY-IS using Barker's algorithm versus $\pi^\lambda$-MALA. Similar to other examples, we run 100 replications for $n = 10^5$ length Markov chains. The box plot in Figure~\ref{fig:boxplot_poisson} presents the average relative efficiencies across components. First, we note the significant gains in efficiency using $\pi^{\lambda}$-HMC. This  corroborates our conjecture that $\pi^{\lambda}$'s heavier tails make it more conducive for HMC to traverse the space. We further note that MY-IS using MALA chains also significantly improves efficiency, compared to both P-MALA and Barker's algorithm, with almost all components exhibiting relative efficiency above 1. 
\begin{figure} 
    \centering
    \includegraphics[scale = .43]{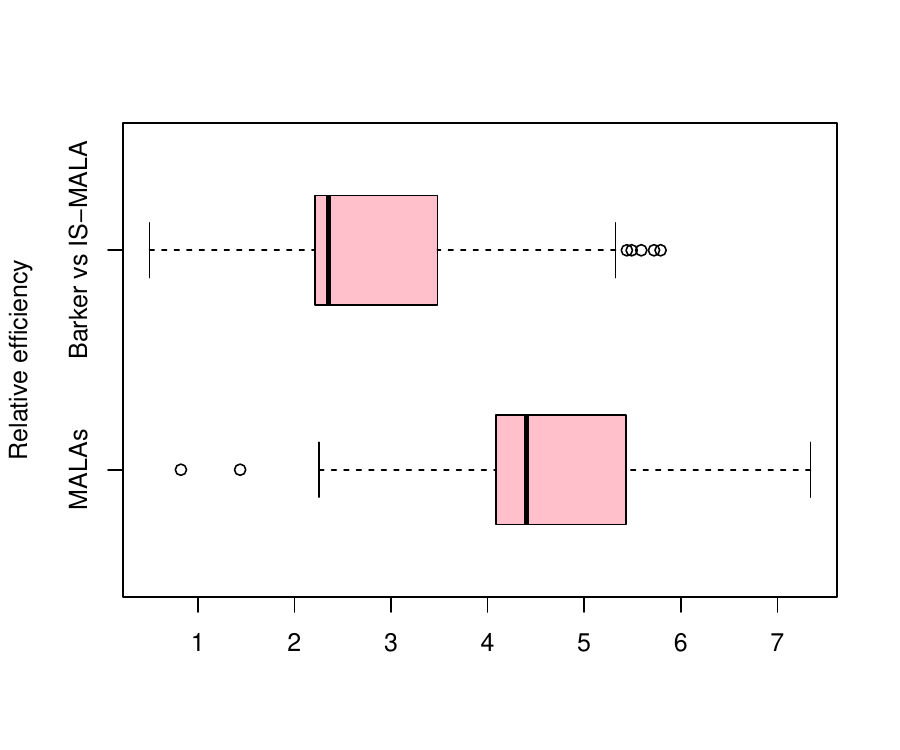}
    \includegraphics[scale = .43]{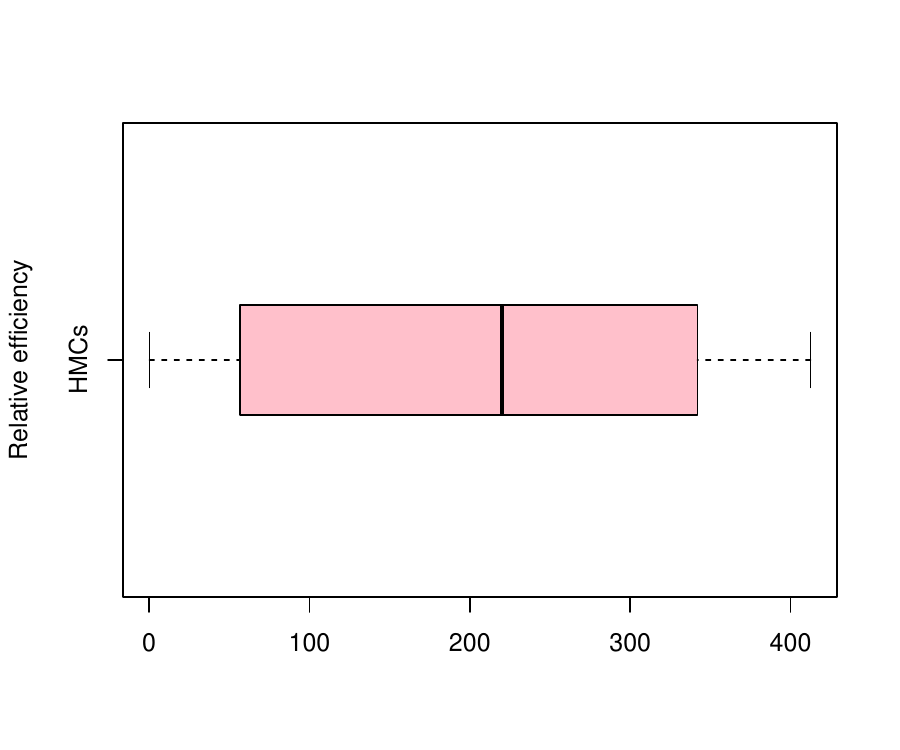}
    \caption{Poisson random effects model: Average relative efficiencies of MYIS-MALA over P-MALA and MYIS-MALA over P-Barker (left) and MYIS-HMC over P-HMC (right).}
\label{fig:boxplot_poisson}
\end{figure}
\section{Discussion}
\label{sec:discussion}

Non-differentiability of a target density is a key obstacle in building effective MCMC strategies.  In this work, we propose an importance sampling paradigm for a class of log-concave targets using Moreau-Yosida envelopes as a proposal. Our proposed estimator is guaranteed to have finite variance, with a Markov chain that can often mix better on the proposal than the original target. We demonstrate the gains in efficiency using our proposed methodology over a variety of examples and highlight the utility of employing $\pi^{\lambda}$-MCMC with importance sampling. 

For sampling from $\pi$, \cite{durmus2022proximal} proposed splitting the potential $\psi$ into a smooth differentiable and a non-differentiable component. This enables enveloping only the non-smooth part of $\psi$. While this strategy could be employed to build an importance sampling scheme, Proposition \ref{prop:my_properties} \ref{max_prop} will no longer be true. As a result, the mode of the importance distribution will not match with the posterior mode, which can have detrimental effects on the weights in high-dimensions. This approach may still be worth pursuing for models where the mismatch in modes is nominal.

Although most applications of importance sampling use iid samples from the importance distribution, MCMC samples from the importance distribution have also been employed with some success \citep[see for example][]{liesenfeld2008improving, schuster2020markov}. However, their use is either in specific problems, low-dimensional settings, or lack guaranteed finite variance of estimators. Adaptive importance sampling methods with MCMC samples are also common in signal processing \citep{martino2018group} but  their application is also typically limited to low-dimensional multi-modal target distributions. Concurrently with our work, \cite
{elvira2024proximal} employ proximal mappings in developing effective adaptive importance sampling paradigms for non-differentiable targets through a proximal Newton adaptation strategy. 

We initially focused on log-concave distributions due to the availability of efficient global solvers for the proximal mapping. Nonetheless, for $\pi$ that is not log-concave, it may be  possible to choose $\lambda$ so that $\pi^{\lambda}$ is log-concave. This would enable our MY-IS estimator to apply to much a wider class of target densities. For instance, similar results are available in the context of weakly convex functions. Specifically, the proximal mapping of a $\rho$-weakly convex function is unique and the gradient of MY envelope is continuously differentiable for all $\lambda \in (0, \rho^{-1})$ \citep{MR4155077, MR4224882}. It is not immediately obvious, however, if $\pi^{\lambda}$ is a proper density in this case -- something that warrants exploring in future work.
The choice of $\lambda$ in this work is crucial and critically dependent on $\pi$ and a more careful analysis of this problem is warranted. We also note that a Moreau-Yosida approximation could be constructed using multiple $\lambda$s, and a methodology similar to parallel tempering can be constructed enabling mode-jumping in a multi-modal target. We leave these problems for future work.
We also defer the study of asymptotic normality of importance sampling quantiles, as such a result would be applicable even outside the scope of proximal MCMC methods. Finally, we note that since \cite{pereyra2016proximal}, alternative approaches to sampling from non-smooth target densities have been introduced \citep{Lee2021, Liang2022, Mou2022}.  We also leave the interesting question of how to potentially adapt our importance sampling framework to these contexts for future work.

\section{Acknowledgements}

Dootika Vats and Eric Chi are grateful to the Rice-IITK Strategic Collaboration Grant for supporting this work while Eric Chi was at Rice University. Dootika Vats is also supported by Google Research.

\appendix

\section{Proof of Theorem~\ref{thm:imp_consistency}}
\label{appendix:consistency_of_snis}

\begin{proof}
For normalising constants $k_{\pi}$ and $k_{g}$, we have
\begin{equation*}
    \pi(x) \amp = \amp \frac{e^{-\psi(x)}}{k_{\pi}}  \quad\quad\text{and}\quad\quad
    g(x) \amp = \amp \frac{\tilde{g}(x)}{k_{g}}\,.
\end{equation*}
From \eqref{eq:est_xi}
\begin{equation*}
   \hat{\theta}^{g}_n \amp = \amp \dfrac{\sum_{t = 1}^{n}\xi(X_{t})w(X_{t})}{\sum_{k = 1}^{n} w(X_{k})} \amp = \amp \dfrac{n^{-1}\sum_{t = 1}^{n}\xi(X_{t})w(X_{t})}{n^{-1}\sum_{k = 1}^{n} w(X_{k})}.
\end{equation*}
Since $w(x) = e^{-\psi(x)}/\tilde{g}(x)$,
\begin{equation*}
    \hat{\theta}^{g}_n \amp = \amp \dfrac{n^{-1}\sum_{t = 1}^{n}\xi(X_{t})\dfrac{k_{\pi}\pi(X_{t})}{k_{g}g(X_{t})}}{n^{-1}\sum_{k = 1}^{n} \dfrac{k_{\pi}\pi(X_{k})}{k_{g}g(X_{k})}} \amp =: \amp \dfrac{Y_n}{Z_n}.
\end{equation*}
By Birkhoff's ergodic theorem \citep[][Section 28.4]{fristedt1996modern}, as $n\to \infty$
\begin{equation*}
    Y_n \amp = \amp n^{-1}\sum_{t = 1}^{n}\xi(X_{t})\dfrac{k_{\pi}\pi(X_{t})}{k_{g} g(X_{t})} \amp \overset{\text{a.s.}}{\to} \amp \frac{k_{\pi}}{k_{g}}\theta\,,
\end{equation*}
and
\begin{equation*}
    Z_n  \amp = \amp  n^{-1}\sum_{k = 1}^{n} \dfrac{k_{\pi}\pi(X_{k})}{k_{g}g(X_{k})} \amp \overset{\text{a.s.}}{\to} \amp \frac{k_{\pi}}{k_{g}}.
\end{equation*}
Therefore, by the continuous mapping theorem $\hat{\theta}^{g}_n\overset{\text{a.s.}}{\rightarrow} \theta$ as $n \to \infty$.
\end{proof}

\section{Proof of Theorem~\ref{thm:asymp_norm}}
\label{app:asymp_norm_proof}

\begin{proof}

It is known that given a $\pi^{\lambda}$-geometrically ergodic Markov chain, a multivariate central limit theorem holds for $\bar{S}_{n} = n^{-1}\sum_{t=1}^{n}S(X_{t})$ if $\E_{\pi^{\lambda}}\|S(X)\|^2 < \infty$ \citep[see][]{MR834478,vats2017output}.

Recall that by \cite{durmus2022proximal}, $\psi^{\lambda}(x) \leq \psi(x)$ for all $x \in \Real^{d}$. Therefore, 
\begin{align*}
    w^{\lambda}(x) \amp = \amp \frac{e^{-\psi(x)}}{e^{-\psi^{\lambda}(x)}} \amp \leq \amp 1\,,
\end{align*}
for all $x \in \Real^{d}$ and hence all moments of $w^{\lambda}(X)$ exist when $X \sim \pi^{\lambda}$. Further, for constants $k_{\pi}$ and $k_{\pi^{\lambda}}$ such that $\pi(x) = e^{-\psi(x)}/k_{\pi}$ and $\pi^{\lambda}(x) = e^{-\psi^{\lambda}(x)}/k_{\pi^{\lambda}}$, and for all $x \in \Real^{d}$,
\begin{equation*} 
    \sup_{x \in \Real^{d}} \frac{\pi(x)}{\pi^{\lambda}(x)} 
\amp = \amp \frac{k_{\pi^{\lambda}}}{k_{\pi}} \sup_{x \in \Real^{d}} \frac{e^{-\psi(x)}}{e^{-\psi^{\lambda}(x)}} \amp \leq \amp \frac{k_{\pi^{\lambda}}}{k_{\pi}} \amp < \amp \infty\,.
\end{equation*}
Given $x \in \Real^{d}$ we write $\xi(x) = \left(
    \xi_{1}(x), \xi_{2}(x), \cdots, \xi_{p}(x)
\right)\Tra$, where $\xi_{i} : \Real^{d} \to \Real$ for all $i = 1, 2, \ldots, p$. Thus,
\begin{align}
    \E_{\pi^{\lambda}}
    |\xi_{i}(X)w^{\lambda}(X)|^{2} \amp &= \amp \int_{\Real^{d}} |\xi_{i}(x)|^{2} |w^{\lambda}(x)|^{2} \pi^{\lambda}(x)\, dx \nonumber \\
   \amp &= \amp \int_{\Real^{d}} |\xi_{i}(x)|^{2} \left(\frac{\exp(-\psi(x))}{\exp(-\psi^{\lambda}(x))}\right)^{2} \pi^{\lambda}(x)\, dx \nonumber \\
   \amp &= \amp \left(\frac{k_{\pi}}{k_{\pi^{\lambda}}}\right)^{2} \int_{\Real^{d}} |\xi_{i}(x)|^{2} \left(\frac{\pi(x)}{\pi^{\lambda}(x)}\right)^{2} \pi^{\lambda}(x)\, dx \nonumber \\
   \amp &= \amp \left(\frac{k_{\pi}}{k_{\pi^{\lambda}}}\right)^{2} \int_{\Real^{d}} |\xi_{i}(x)|^{2} \frac{\pi(x)}{\pi^{\lambda}(x)} \pi(x)\, dx \nonumber \\
   \amp &\leq \amp \left(\frac{k_{\pi}}{k_{\pi^{\lambda}}}\right)^{2} \sup_{x \in \Real^{d}} \frac{\pi(x)}{\pi^{\lambda}(x)} \int_{\Real^{d}} |\xi_{i}(x)|^{2} \pi(x)\, dx \nonumber \\
   \amp &\leq \amp \left(\frac{k_{\pi}}{k_{\pi^{\lambda}}}\right) \int_{\Real^{d}} |\xi_{i}(x)|^{2} \pi(x)\, dx \nonumber \\
   \amp &= \amp \left(\frac{k_{\pi}}{k_{\pi^{\lambda}}}\right) \E_{\pi}(|\xi_{i}(X)|^2)\,. \nonumber
\end{align}
Since $\E_{\pi}\lVert\xi(X)\rVert^{2} < \infty$, $\E_{\pi^{\lambda}}
    |\xi_{i}(X)w^{\lambda}(X)|^{2}$ is finite. Thus, $\E_{\pi^{\lambda}}\|S(X)\|^2 < \infty$. By a Markov chain central limit theorem, there exists a $(p+1) \times (p+1)$ positive-definite matrix $\Sigma$ such that,
\begin{equation} \label{vector_CLT}
    \sqrt{n}\left(\bar{S}_{n} - \E_{\pi^{\lambda}} (S(X)) \right) \overset{d}{\to}  N_{p+1}(0, \Sigma)\,,
\end{equation}
where,
\begin{equation*}
    \Sigma \amp = \amp   \Var_{\pi^{\lambda}}S(X_{1}) + \sum_{k=1}^{\infty}\Cov_{\pi^{\lambda}}(S(X_{1}), S(X_{1+k})) + \sum_{k=1}^{\infty}\Cov_{\pi^{\lambda}}(S(X_{1+k}), S(X_{1}))\,.
\end{equation*}
Let $\kappa: \Real^{p+1} \to \Real^{p}$ be a continuous real-valued function such that for $u \in \Real^{p}$ and $v \in \Real$,
\begin{equation} \label{kappa_func}
   \kappa \begin{pmatrix}
        u \\
        v
    \end{pmatrix}  \amp = \amp \frac{u}{v}.
\end{equation}
Using multivariate delta method in \eqref{vector_CLT} for function $\kappa$, as $n \to \infty$,
\begin{equation} \label{mult_delta_eqn}
    \sqrt{n}\left( \kappa (\bar{S}_{n})  - \kappa (\E_{\pi^{\lambda}}S(X)) \right) \overset{d}{\to} N_{p}(0, \nabla \kappa_{\eta} \Sigma \nabla \kappa_{\eta}\Tra)\,,
\end{equation}
where $\nabla \kappa_{\eta}$ is the total derivative of $\kappa$ at the point $ \eta = \E_{\pi^{\lambda}}[S(X)]$. A sufficient condition for the existence of the total derivative is that all the partial derivatives exist in a neighbourhood of $\eta$ and are continuous at $\eta$ \citep{van2000asymptotic}. From \eqref{kappa_func}, writing $u = \left(
   u_{1},  u_{2},  \cdots, u_{p}
\right)\Tra$ where $u_{i} \in \Real$ for all $i = 1, 2, \ldots, p$ we have $\kappa \left(
    u, v
\right)\Tra = \left(
    u_{1}/v, u_{2}/v, \cdots, u_{p}/v
\right)\Tra$. The total derivative is then just the $p \times (p+1)$ matrix of partial derivatives given by
\begin{align*}
    \nabla \kappa_{\left(
    u, v
\right)\Tra} \amp &= \amp \begin{bmatrix}
        \dfrac{\textbf{I}_{p}}{v} & -\dfrac{u}{v^{2}}
    \end{bmatrix}.
\end{align*}
Therefore,
    \begin{align*}
    \nabla \kappa_{\eta} \amp &= \amp \begin{bmatrix}
        \dfrac{\textbf{I}_{p}}{\E_{\pi^{\lambda}}(w^{\lambda}(X))} & - \dfrac{\E_{\pi^{\lambda}}(\xi(X)w^{\lambda}(X))}{(\E_{\pi^{\lambda}}(w^{\lambda}(X)))^{2}}
    \end{bmatrix} \nonumber \\ 
   \amp &= \amp \dfrac{1}{\E_{\pi^{\lambda}}(w^{\lambda}(X))} \begin{bmatrix}
        \textbf{I}_{p} & -\theta
    \end{bmatrix}.
\end{align*}
Thus,
\begin{align}
    \nabla \kappa_{\eta} \Sigma \nabla \kappa_{\eta}\Tra \amp &= \amp \dfrac{1}{(\E_{\pi^{\lambda}}(w^{\lambda}(X)))^{2}}
     \begin{bmatrix}
        \textbf{I}_{p} & -\theta
    \end{bmatrix} \Sigma \begin{bmatrix}
        \textbf{I}_{p} \\
        -\theta\Tra
    \end{bmatrix} \amp = \amp \Xi\,.
    \label{asymp_variance}
\end{align}
Since $\hat{\theta}_{n}^{\text{MY}} = \kappa (\bar{S}_{n})$ and $\theta = \kappa (\E_{\pi^{\lambda}}S(X))$, using \eqref{asymp_variance} in \eqref{mult_delta_eqn} we have as $n \to \infty$,
    \begin{equation*}
        \sqrt{n}(\hat{\theta}_{n}^{\text{MY}} - \theta) \overset{d}{\to} N_{p}(0, \Xi)\,.
    \end{equation*}
\end{proof}

\section{Estimation of asymptotic variance}
\label{sec:estimating_variance}

Theorem~\ref{thm:asymp_norm} guarantees that, subject to convergence rates of the $\pi^{\lambda}$-Markov chain, the MY-IS estimator has a finite covariance matrix $\Xi$. Practitioners would further require an estimator of $\Xi$ in order to assess the Monte Carlo error in estimation. This involves estimating $\Sigma$ in,
\begin{eqnarray}
\label{eq:true_lambda_supp}
    \Xi  & = &  \dfrac{1}{ \left[\E_{\pi^{\lambda}}(w^{\lambda}(X_1)) \right]^{2}}
     \begin{bmatrix}
        \textbf{I}_{p} & -\theta
    \end{bmatrix} \Sigma \begin{bmatrix}
        \textbf{I}_{p} \\
        -\theta\Tra
    \end{bmatrix}\,,
\end{eqnarray}
 the asymptotic covariance matrix in the Markov chain central limit theorem for the process $\{S(X_t)\}_{t\geq 1}$. A number of estimators of the asymptotic covariance are available in the literature. \cite{chen:seila:1987} and \cite{vats2019multivariate} employ a batch means estimator and demonstrate its strong consistency. We employ this estimator due to its computational efficiency and well established asymptotic properties.

Let $n = ab$ where $a$ denotes the number of batches and $b$, the size of a batch. Let $\bar{T}_{k} = b^{-1} \sum_{j=1}^{b} S(X_{kb + j})$ for $k = 0, 1, \ldots, a-1$  be the mean vector of the $k^{\text{th}}$ batch, and $\bar{S}_{n} = n^{-1}\sum_{t=1}^{n} S(X_{t})$ be the overall mean. The batch means estimator of $\Sigma$ is defined as
\begin{eqnarray*}
    \hat{\Sigma}_n & = & \frac{b}{a-1} \sum_{k=0}^{a-1} (\bar{T}_{k} - \bar{S}_{n})(\bar{T}_{k} - \bar{S}_{n})^{\Tra} \,.
\end{eqnarray*}
Using $\hat{\Sigma}_n$, a plug-in estimator of $\Xi$ can be constructed. Denote $\bar{w}_n  =  n^{-1}\sum_{t=1}^{n} w^{\lambda}(X_t)\,.$
%
%
Then a plug-in estimator of $\Xi$ is
\begin{eqnarray}
    \label{eq:lambda_est}
    \hat{\Xi}^\text{BM}_n & = & \dfrac{1}{ \bar{w}^2_n}
     \begin{bmatrix}
        \textbf{I}_{p} & - \hat{\theta}^{\text{MY}}_n
    \end{bmatrix} \hat{\Sigma}_n \begin{bmatrix}
        \textbf{I}_{p} \\
        - \left(\hat{\theta}^{\text{MY}}_n\right)\Tra
    \end{bmatrix}\,.
\end{eqnarray}

Under the strong consistency conditions for $\hat{\Sigma}_n$ discussed in \cite{vats:flegal:2018,vats2019multivariate} and using the continuous mapping theorem, $\hat{\Xi}^\text{BM}_n$ can also be shown to be strongly consistent. Alternative estimators of $\Sigma$ exist that may better suit a user's preferences. Spectral variance estimators \citep{vats:fleg:jon:2018}, regenerative estimators \citep{seil:1982},  moment least squares estimators \citep{berg2023efficient,song2024multivariate}, and initial sequence estimators \citep{banerjee2024efficient,dai:jon:2017,gey:1992} are all well-studied with conditions for strong consistency available. See \cite{flegal2024implementing} for a  review. We thus present the following result generally for any estimator of $\Sigma$.
\begin{theorem}
    Let $\hat{\Sigma}$ be strongly consistent for $\Sigma$, and let $\hat{\Xi}_n$ be the estimator of $\Xi$ constructed using $\hat{\Sigma}$. Then, $\hat{\Xi}_n \overset{\text{a.s.}}{\to} \Xi$ as $n\to \infty$.
\end{theorem}

\begin{proof}
    The result follows from the continuous mapping theorem and the fact that both $\hat{\theta}^{\text{MY}}_n$ and $\bar{w}_n$ are strongly consistent.
\end{proof}

\section{Optimal \texorpdfstring{$\lambda$}{lambda} for Gaussian target}
\label{app: opt_lamb_gaussian}
\begin{proof}[Proof of Theorem~\ref{thm:my_normal}]
We  first find the form of the MY envelope of $\psi$ for $N(0, \Omega)$ target. Note that
\begin{equation*}
    \psi^{\lambda}(x) \amp = \amp \min_{y \in \mathbb{R}^p} \left\{ \dfrac{y\Tra \Omega^{-1} y}{2} + \dfrac{(x - y)\Tra(x-y)}{2 \lambda}  \right\}\,.
\end{equation*}
It is then easy to see that $\text{prox}_{\psi}^{\lambda}(x) \amp = \amp (\lambda \Omega^{-1} + \mathbb{I}_d)^{-1} x$. Therefore,
\begin{align}
    \psi^{\lambda}(x) &\amp = \amp 
    \dfrac{\text{prox}_{\psi}^{\lambda}(x)\Tra \Omega^{-1} \text{prox}_{\psi}^{\lambda}(x)}{2} + \dfrac{(x - \text{prox}_{\psi}^{\lambda}(x))\Tra(x-\text{prox}_{\psi}^{\lambda}(x))}{2 \lambda} \label{eq: MY_proxsum}\,.
\end{align}
The following two simplifications will be used later
\begin{align*} \dfrac{\text{prox}_{\psi}^{\lambda}(x)\Tra \Omega^{-1} \text{prox}_{\psi}^{\lambda}(x)}{2} \amp &= \amp \dfrac{x\Tra(\lambda \Omega^{-1} + \mathbb{I}_d)^{-1}  \Omega^{-1} (\lambda \Omega^{-1} + \mathbb{I}_d)^{-1} x}{2} \quad \text{and} \\
    \dfrac{(x - \text{prox}_{\psi}^{\lambda}(x))\Tra(x-\text{prox}_{\psi}^{\lambda}(x))}{2 \lambda} \amp &= \amp \dfrac{x\Tra(\mathbb{I}_d - (\lambda \Omega^{-1} + \mathbb{I}_d)^{-1} )\Tra(\mathbb{I}_d - (\lambda \Omega^{-1} + \mathbb{I}_d)^{-1})x }{2\lambda}\,. 
\end{align*}
Using Woodbury identity,
\begin{equation*}
    (\lambda \Omega^{-1} + \mathbb{I}_d)^{-1} \amp = \amp \mathbb{I}_d - (\mathbb{I}_d + \lambda \Omega^{-1})^{-1}\lambda \Omega^{-1} \amp = \amp \mathbb{I}_d - \lambda \Omega^{-1}(\mathbb{I}_d + \lambda \Omega^{-1})^{-1}\,.
\end{equation*}
Thus, we obtain
\begin{equation*} 
     \dfrac{(x - \text{prox}_{\psi}^{\lambda}(x))\Tra(x-\text{prox}_{\psi}^{\lambda}(x))}{2 \lambda} \amp = \amp \dfrac{x\Tra(\lambda \Omega^{-1} + \mathbb{I}_d)^{-1}  \lambda \Omega^{-1} \Omega^{-1}(\lambda \Omega^{-1} + \mathbb{I}_d)^{-1}x }{2}\,. 
\end{equation*}
So in \eqref{eq: MY_proxsum} we have
\begin{align*}
    \psi^{\lambda}(x) \amp &= \amp \dfrac{x\Tra(\lambda \Omega^{-1} + \mathbb{I}_d)^{-1} \Omega^{-1} (\lambda \Omega^{-1} + \mathbb{I}_d)^{-1} x}{2} \\
      & \qquad  + \dfrac{x\Tra(\lambda \Omega^{-1} + \mathbb{I}_d)^{-1}  \lambda \Omega^{-1} \Omega^{-1}(\lambda \Omega^{-1} + \mathbb{I}_d)^{-1}x }{2} \\
        \Rightarrow \quad \psi^{\lambda}(x) \amp & = \amp \dfrac{x\Tra(\Omega + \lambda\mathbb{I}_d)^{-1} x}{2}\,.
      \end{align*}
Since normalizing constants are unique, the MY envelope of $N(0, \Omega)$ is $N(0, (\Omega + \lambda \mathbb{I}_d) )$.

\cite{agarwal2022principled} provided the form of the asymptotic covariance matrix of the self-normalized importance sampling estimator of $\E_{\pi}(X)$ when using iid samples from $\pi^{\lambda}$. Using their result, we obtain that the limiting covariance matrix for the (iid\@) self-normalized importance sampling estimator for $\xi(x) = x$ is,
\begin{equation}
\label{eq:gauss_lambda}
   \Xi \amp := \amp \dfrac{|\Omega + \lambda \mathbb{I}_d|^{1/2} }{ |\Omega|  \, | 2\Omega^{-1} -(\Omega + \lambda \mathbb{I}_d)^{-1}|^{1/2}}  \left(2\Omega^{-1} -(\Omega + \lambda \mathbb{I}_d)^{-1}  \right)^{-1}.
\end{equation}
We may then choose $\lambda$ based on the value that minimizes $|\Xi|$. The rest of the argument finds this optimal value of $\lambda$. First, consider the eigenvalue decomposition of $\Omega = LDL^T$ where $L$ is the orthogonal matrix of the eigenvectors of $\Omega$ and $D = \text{diagonal}(s_1, s_2, \dots, s_d)$ where $s_i$ is the $i$\textsuperscript{th} eigenvalue of $\Omega$. Using this,
\begin{align*}
    2\Omega^{-1} -(\Omega + \lambda \mathbb{I}_d)^{-1} \amp & = \amp 2 LD^{-1} L - (L D L\Tra + \lambda LL\Tra)^{-1}\\
    \amp & = \amp L \left(2D^{-1} - (D + \lambda \mathbb{I}_d)^{-1} \right) L\Tra\\
    \amp & = \amp L \,\text{diag}(w_{1}, w_{2}, \ldots, w_{p})\, L\Tra\,, \numberthis \label{eq:gauss-form}
\end{align*}
where $w_{i} = (s_i + 2\lambda)/(s_i(s_i + \lambda))$. Using \eqref{eq:gauss-form} in \eqref{eq:gauss_lambda},
\begin{align}
    |\Xi| \amp &= \amp  \left |  \dfrac{|\Omega + \lambda \mathbb{I}_d|^{1/2} }{ |\Omega|  \, | 2\Omega^{-1} -(\Omega + \lambda \mathbb{I}_d)^{-1}|^{1/2}}  \left(2\Omega^{-1} -(\Omega + \lambda \mathbb{I}_d)^{-1}  \right)^{-1}  \right| \nonumber\\
    \amp & = \amp \left(\dfrac{|\Omega + \lambda \mathbb{I}_d|^{1/2} }{ |\Omega|  \, | 2\Omega^{-1} -(\Omega + \lambda \mathbb{I}_d)^{-1}|^{1/2}} \right)^{d}  \left| \left(2\Omega^{-1} -(\Omega + \lambda \mathbb{I}_d)^{-1}  \right)^{-1} \right| \nonumber\\
    \Rightarrow \log |\Xi| \amp & = \amp \dfrac{d}{2} \log |\Omega + \lambda \mathbb{I}_d| - d\log |\Omega|- \dfrac{d}{2} \log | 2\Omega^{-1} -(\Omega + \lambda \mathbb{I}_d)^{-1}| \nonumber\\
    & \hspace{7cm} - \log | 2\Omega^{-1} -(\Omega + \lambda \mathbb{I}_d)^{-1}| \nonumber\\
    \amp & = \amp d\sum_{i=1}^{d} \left[ \dfrac{1}{2}\log (s_i + \lambda) -\log s_i - \dfrac{1}{2}\log  \left(\dfrac{s_i + 2\lambda}{s_i(s_i + \lambda)} \right) \right] \nonumber\\ 
    & \hspace{7cm} - \sum_{i=1}^{d} \log \left(\dfrac{s_i + 2\lambda}{s_i(s_i + \lambda)} \right) \nonumber\\ 
    \amp & = \amp \sum_{i=1}^{d} \left( (d+1) \log(s_i + \lambda) - \left(\dfrac{d}{2} - 1 \right)\log s_i - \left( \dfrac{d}{2} + 1 \right)\log(s_i + 2\lambda) \right) \nonumber\\
    \Rightarrow \dfrac{d \log |\Xi|}{d\lambda} \amp & = \amp \sum_{i=1}^{d} \left(\dfrac{d+1}{s_i + \lambda} - \dfrac{d + 2}{s_i + 2\lambda} \right) \amp \overset{\text{set}}{=}\amp 0 \nonumber\\
    \amp & \Rightarrow \amp \sum_{i=1}^{d} \dfrac{\lambda d - s_i}{ (s_i + \lambda) (s_i + 2 \lambda)}\amp = \amp 0\,.\label{eq:first_deriv_condn}
\end{align}
We check the existence of a solution to \eqref{eq:first_deriv_condn} for $\lambda > 0$. Denote $g(\lambda) := \log \lvert \Xi \rvert$, then, $g(\lambda)$ decreases on $[0, s_{1}/d)$ since $g'(\lambda) < 0$ in this interval. Thus, $g(\lambda) > g(s_{1}/d)$ for all $\lambda \in [0, s_{1}/d)$. Similarly $g(\lambda)$ increases on $(s_{d}/d, \infty)$ since $g'(\lambda) > 0$ on this interval. Thus, $g(\lambda) > g(s_{d}/d)$ for all $\lambda \in (s_{d}/d, \infty)$.
Since $g(\lambda)$ is continuous, it attains a global minimum over the compact set $\left[s_{1}/d, s_{d}/d\right]$, i.e., $g(\lambda) \geq g(\lambda^*)$ where $\lambda^* \in \left[s_{1}/d, s_{d}/d\right]$.
But $g(\lambda^*) \leq \min\{g(s_{1}/d), g(s_{d}/d)\}$. Therefore, $\lambda^*$ minimizes $g(\lambda)$ over $\lambda > 0$ and $s_1/d \leq \lambda^* \leq s_d/d$.
\end{proof}
\subsection{Importance sampling ESS in iid Gaussian case}

As before, let the target density $\pi$ be that of $N(0, \Omega)$. Further suppose all the eigenvalues are the same so that $\lambda^* = s_1/d$. Thus, $\pi^{\lambda^*}$ is the density of $N(0, W)$, where $W = \Omega + (s_1/d)\mathbb{I}_d$\,. The importance sampling ESS of \cite{kong1992note} is,
\begin{equation} \label{eq: kong_ess_formula}
    \dfrac{n_e}{n} \amp = \amp \dfrac{\bar{w}_n^{2}}{\,\,\overline{w^{2}}_n} \amp \approx \amp \dfrac{(\E_{\pi^{\lambda}}(w^{\lambda}(X)))^2}{\E_{\pi^{\lambda}}(w^{\lambda}(X)^2)} \,.
\end{equation}
Note that,
\begin{align} \label{eq: asymp_imp_ess_num}
    \E_{\pi^{\lambda}}(w^{\lambda}(X)) \amp = \amp \int_{\Real^d} \frac{\pi(x)}{k_{\pi}} \frac{k_{\pi^{\lambda}}}{\pi^\lambda(x)} \pi^{\lambda}(x) dx \amp = \amp \frac{k_{\pi}}{k_{\pi^{\lambda}}}\,. 
\end{align}
Further,
\begin{align*}
    \E_{\pi^{\lambda}}(w^{\lambda}(X)^2) \amp &= \amp \int_{\Real^d} \frac{\pi(x)^2}{k_{\pi}^2} \frac{k_{\pi^{\lambda}}^2}{\pi^\lambda(x)^2} \pi^{\lambda}(x) \,dx \\
    \amp &= \amp \frac{k_{\pi}^2}{k_{\pi^{\lambda}}^2} \int_{\Real^d} \frac{1}{(\sqrt{2\pi})^d} \frac{|W|^{1/2}}{|\Omega|} \exp(-2\psi(x) + \psi^{\lambda}(x)) \,dx \\
    \amp &= \amp \frac{k_{\pi}^2}{k_{\pi^{\lambda}}^2} \int_{\Real^d} \frac{1}{(\sqrt{2\pi})^d} \frac{|W|^{1/2}}{|\Omega|} \exp\left\{-x\Tra\Omega^{-1}x + \frac{1}{2}x\Tra W^{-1}x\right\}dx \\
    \amp &= \amp \frac{k_{\pi}^2}{k_{\pi^{\lambda}}^2} \int_{\Real^d} \frac{1}{(\sqrt{2\pi})^d} \frac{|W|^{1/2}}{|\Omega|} \exp\left\{-\frac{1}{2}\left[x\Tra\left(2\Omega^{-1} - W^{-1}\right)x\right]\right\}dx \\
    \amp &= \amp \frac{k_{\pi}^2}{k_{\pi^{\lambda}}^2} \frac{|W|^{1/2}}{\Omega} |(2\Omega^{-1} - W^{-1})^{-1}|^{1/2} \;\; \times \\
    & \,\,\,  \int_{\Real^d} \frac{1}{(\sqrt{2\pi})^d} \frac{1}{|(2\Omega^{-1} - W^{-1})^{-1}|^{1/2}} \exp\left(-\frac{1}{2}x\Tra (2\Omega^{-1} - W^{-1})x\right)dx\,.
\end{align*}
It is known that given two conformable matrices $A$ and $B$, a sufficient condition for $A - B$ to be positive definite is that the minimum eigenvalue of $A$ is greater than the maximum eigenvalue of $B$. Thus, $(2\Omega^{-1} - W^{-1})$ is positive definite, and the integral on the right hand side in the above equation is 1. Therefore,
\begin{equation} \label{eq: asymp_imp_ess_denom}
    \E_{\pi^{\lambda}}(w^{\lambda}(X)^2) \amp = \amp \frac{k_{\pi}^2}{k_{\pi^{\lambda}}^2} \frac{|W|^{1/2}}{|\Omega|} |(2\Omega^{-1} - W^{-1})^{-1}|^{1/2}\,.
\end{equation}
 Using \eqref{eq: asymp_imp_ess_num} and \eqref{eq: asymp_imp_ess_denom} in \eqref{eq: kong_ess_formula},
 \begin{equation} \label{eq: n_e/n_form}
     \dfrac{n_e}{n} \amp = \amp \dfrac{|2\Omega^{-1} - W^{-1}|^{1/2}|\Omega|}{|W|^{1/2}} \,.
 \end{equation}
Since $\Omega$ is a $d \times d$ diagonal matrix with same eigenvalues, namely $\Omega := \text{diag}(s_1)_{d \times d}$\,,
\begin{align*}
    2\Omega^{-1} - W^{-1} \amp = \amp \text{diag} \left( \frac{2}{s_1} - \frac{1}{s_1 + s_1/d} \right) \amp = \amp \text{diag} \left(\frac{1}{s_1} \left(\frac{d+2}{d+1}\right) \right)_{d \times d}\,.
\end{align*}
Noting that $|W|^{1/2} = s_1^{d/2}((d+1)/d)^{d/2}$ and substituting from the last equation into \eqref{eq: n_e/n_form} we have,
\begin{equation*}
    \dfrac{n_e}{n} \amp = \amp \dfrac{(d(d+2))^{d/2}}{(d+1)^d} \amp = \amp \dfrac{(1 + 2/d)^{d/2}}{(1 + 1/d)^{d}}\,.
\end{equation*}
Now both numerator and denominator converge to $e$ as $d \to \infty$. Thus, $n_e / n \to 1$ as $d \to \infty$.

\section{Proofs of geometric ergodicity of \texorpdfstring{$\pi^{\lambda}$}{pilambda}-Markov chains}
\label{app:proof_pi-lambda_chains}

We will require the following lemma in both the proofs of geometric ergodicity for $\pi^{\lambda}$-MALA and $\pi^{\lambda}$-HMC.

\begin{lemma}
\label{lem:x_1-prox_limit}
Under Assumption~\ref{ass:limsup},
\begin{equation*}
    \liminf_{\|x\| \to \infty} \left\{\|x\| \left(1 - \dfrac{\|\prox^{\lambda}_{\psi}(x)\|}{\|x\|}  \right)\right\} \amp = \amp  \infty \,.
\end{equation*}
\end{lemma}

\begin{proof}
Let $B(t) := \{x: \|x\| > t\}$.  
Assumption~\ref{ass:limsup} is equivalent to 
\begin{eqnarray}
\label{eq: limsup_assum1}
\inf_{t\geq 0} \sup_{x \in B(t)}  \dfrac{\|\prox^{\lambda}_{\psi}(x)\|}{\|x\|} & = & l\,.
\end{eqnarray}
Equation \eqref{eq: limsup_assum1} implies that for any positive $\epsilon$ there exists $t_\epsilon \geq 0$ such that
\begin{eqnarray*}
\dfrac{\|\prox^{\lambda}_{\psi}(x)\|}{\|x\|} & < & l + \epsilon \qquad  \text{for all } x \in B(t_\epsilon)\,.
\end{eqnarray*}
Consider a sequence  $\epsilon_n = 1/n$. Then there exists an increasing and diverging sequence $t_n$ such that
\begin{eqnarray}
\label{eq:liminf_bound}
\lVert x \rVert \left(1 - \dfrac{\|\prox^{\lambda}_{\psi}(x)\|}{\|x\|}\right) & > & \lVert x \rVert \left(1 - l -  \frac{1}{n}\right) \quad \text{for all } x \in B(t_n)\,.
\end{eqnarray}
We will show the existence of such a sequence towards the end of this proof. Taking the infimum of both sides of the inequality \eqref{eq:liminf_bound} over $x \in B(t_n)$,
\begin{eqnarray*}
\underset{x \in B(t_n)}{\inf}\;\left\{
\lVert x \rVert \left(1 - \dfrac{\|\prox^{\lambda}_{\psi}(x)\|}{\|x\|}\right)\right\} & \geq & t_n \left(1 - l - \frac{1}{n} \right).
\end{eqnarray*}
Note that for $t \geq t_n$, $B(t) \subset B(t_n)$. Consequently,
\begin{eqnarray}
\label{eq:liminf_bound2}
\underset{x \in B(t)}{\inf}\;\left\{
\lVert x \rVert \left(1 - \dfrac{\|\prox^{\lambda}_{\psi}(x)\|}{\|x\|}\right)\right\} & \geq & t_n \left(1 - l - \frac{1}{n} \right)\,.
\end{eqnarray}
Taking the limit as $t \rightarrow \infty$ on both sides of the  inequality in \eqref{eq:liminf_bound2} gives us
\begin{eqnarray}
\label{eq:liminf_bound3}
\underset{t \rightarrow \infty}{\lim}
\underset{x \in B(t)}{\inf}\;\left\{
\lVert x \rVert \left(1 - \dfrac{\|\prox^{\lambda}_{\psi}(x)\|}{\|x\|}\right)\right\} & \geq & t_n \left(1 - l - \frac{1}{n} \right).
\end{eqnarray}
Since $t_n$ is an increasing and diverging sequence, taking the limit as $n \rightarrow \infty$ of both sides of the 
inequality in \eqref{eq:liminf_bound3}  gives the desired result.

What is left to show is that there exists an increasing and diverging sequence $t_n$.
Note that there exists $\tilde{t}_{n+1}$ such that
\begin{eqnarray*}
\dfrac{\|\prox^{\lambda}_{\psi}(x)\|}{\|x\|} & < & l + \frac{1}{n+1} \qquad \text{for all } x \in B(\tilde{t}_{n+1})\,.
\end{eqnarray*}
Let $t_{n+1} := \max(\tilde{t}_{n+1}, t_n + 1) \geq t_n + 1 > t_n$. But since $t_{n+1} \geq \tilde{t}_{n+1}$, we have that $B(t_{n+1}) \subset B(\tilde{t}_{n+1})$. Therefore,
\begin{eqnarray*}
\dfrac{\|\prox^{\lambda}_{\psi}(x)\|}{\|x\|} & < & l + \frac{1}{n+1} \qquad  \text{for all }x \in B(t_{n+1})\,.
\end{eqnarray*}
\end{proof}

\subsection{Proof of Theorem~\ref{thm:geom_mala}} \label{app:proof_mala}
\begin{proof}
We need to show that when $h \leq 2\lambda$
    \begin{eqnarray*}
        \eta & := & \liminf_{\|x\| \rightarrow \infty} \left\{\|x\| - \|c(x)\|\right\} \amp > \amp 0\,.
     \end{eqnarray*}
 Recall that $\nabla \log \pi^{\lambda}(x) = - \lambda^{-1}\left[x - \prox^{\lambda}_{\psi}(x)\right]$.
    Therefore,
    \begin{align*}
        \|x\| - \|c(x)\| \amp &= \amp \|x\| - \left\|x + \frac{h}{2}\nabla \log \pi^{\lambda}(x)\right\| \\
       \amp &= \amp \|x\| - \left\|x - \frac{h}{2\lambda}\left[x - \prox^{\lambda}_{\psi}(x)\right]\right\| \\
       \amp &= \amp \|x\| - \left\|\left(1 - \frac{h}{2\lambda}\right)x + \frac{h}{2\lambda}\prox^{\lambda}_{\psi}(x)\right\| \\
       \amp & \geq \amp \|x\| - \left|1 - \frac{h}{2\lambda}\right|\|x\| - \frac{h}{2\lambda}\left\|\prox^{\lambda}_{\psi}(x)\right\|\,,
    \end{align*}
    where the last inequality follows from the triangle inequality.
   Since $h/(2\lambda) \leq 1$,
    \begin{align*}
       \liminf_{\|x\| \to \infty}  \left\{ \|x\| - \|c(x)\| \right\} \amp &\geq \amp \liminf_{\|x\| \to \infty} \left \{\|x\| - \left(1 - \frac{h}{2\lambda}\right)\|x\| - \frac{h}{2\lambda}\left\|\prox^{\lambda}_{\psi}(x)\right\| \right\} \\
        \amp &= \amp \liminf_{\|x\| \to \infty} \left \{\frac{h}{2\lambda}(\|x\| - \left\|\prox^{\lambda}_{\psi}(x)\right\|)\right\} \\
        \amp &= \amp \frac{h}{2\lambda} 
        \liminf_{\|x\| \to \infty} \left \{ \|x\| \left(1 -  \frac{\left\|\prox^{\lambda}_{\psi}(x)\right\|}{\|x\|}\right)\right\} \\
       \amp &> \amp 0\,,
    \end{align*}
    where the last inequality follows from Lemma~\ref{lem:x_1-prox_limit}.
\end{proof}

\subsection{Proof of Theorem~\ref{thm:geom_hmc}}
\label{app:proof_hmc}

Condition \eqref{limsup_condn} can be challenging to verify, and hence \cite{MR4003576} provide the following sufficient conditions for it to hold.
\begin{theorem} [\cite{MR4003576}] 
\label{thm:suff_condn_hmc}
    For any $L \geq 1$, \eqref{limsup_condn} holds if the following are met
    \begin{align}
        &\text(a) \qquad \lim_{\|x\| \rightarrow \infty} \|\nabla \psi(x)\| \amp = \amp \infty\,, \label{eq:living1} \\
        &\text(b) \qquad \liminf_{\|x\| \rightarrow \infty} \frac{\langle \nabla \psi(x), x \rangle}{\|\nabla \psi(x)\| \|x\|} \amp > \amp 0\,,\label{eq:living2} \\
        &\text(c) \qquad \lim_{\|x\| \rightarrow \infty} \frac{\|\nabla \psi(x)\|}{\|x\|} \amp = \amp 0\,. \label{eq:living3}
        \end{align}
        In addition, if ($c$) is replaced by, 
        \vspace{-0.5cm}
        \begin{align}
        \hspace{-.5cm}\text(c^*) \qquad \limsup_{\|x\| \rightarrow \infty} \frac{\|\nabla \psi(x)\|}{\|x\|} \amp = \amp S_{l}\,, \label{eq:living3also}
    \end{align}
    for some $S_{l} < \infty$, then there exists an $\varepsilon_{0} < \infty$ such that \eqref{limsup_condn} holds provided $\varepsilon < \varepsilon_{0}$.
\end{theorem}
As it turns out, under Assumption~\ref{ass:limsup}, \eqref{eq:living1}, \eqref{eq:living2}, and \eqref{eq:living3also} hold for $\pi^{\lambda}$, and verifying geometric ergodicity is significantly simpler. 

\begin{proof}
    We will show that under Assumption~\ref{ass:limsup}, conditions \eqref{eq:living1}, \eqref{eq:living2}, and \eqref{eq:living3also} hold, yielding the result. First,
    \begin{align*}
    \lim_{\|x\| \rightarrow \infty} \|\nabla \psi^{\lambda}(x)\| \amp & = \amp \lim_{\|x\| \rightarrow \infty} \dfrac{\|x - \prox_{\psi}^{\lambda}(x)\|}{\lambda} \\
    \amp  & \geq \amp \dfrac{1}{\lambda}\lim_{\|x\| \rightarrow \infty}\|x\|\left(1 - \dfrac{\|\prox_{\psi}^{\lambda}(x)\|}{\|x\|} \right) = \amp \infty\,, \numberthis \label{eq:lim_grad_psi_lam_infty}
    \end{align*}
using Lemma~\ref{lem:x_1-prox_limit} and the fact that for any sequence $a_n$, $\liminf a_n \leq \limsup a_n$. Thus, \eqref{eq:living1} holds.

We move on to showing \eqref{eq:living2} holds. Let $\psi^*$ denote the Fenchel conjugate of $\psi$, i.e.,
$$
\psi^*(x) \amp = \amp \underset{y}{\sup}\; \left\{\langle x, y \rangle - \psi(y) \right\}.
$$
Recall that $\psi^* \in \Gamma(\Real^d)$  whenever 
$\psi \in \Gamma(\Real^d)$. Furthermore, recall that if $\psi \in \Gamma(\Real^d)$ then by the Moreau decomposition, $x \amp = \amp \prox^{\lambda}_{\psi}(x) + \lambda \prox^{1/\lambda}_{\psi^*}(x/\lambda)$, and thus
\begin{equation*}
    \nabla \psi^{\lambda}(x) \amp = \amp \frac{1}{\lambda}[x - \prox^{\lambda}_{\psi}(x)] \amp = \amp \prox^{1/\lambda}_{\psi^*}(x/\lambda)\,.
\end{equation*}
Since $\psi^* \in \Gamma(\Real^d)$, the proximal mapping is firmly nonexpansive, i.e., for all $x$ and $y$
\begin{equation}
\label{eq:firm_nonexpansive}
    \lVert \prox^{1/\lambda}_{\psi^*}(x) - \prox^{1/\lambda}_{\psi^*}(y) \rVert^2 \amp \leq \amp \langle x - y, \prox^{1/\lambda}_{\psi^*}(x) - \prox^{1/\lambda}_{\psi^*}(y) \rangle\,.
\end{equation}
Plugging $(x, y) = \left (x/\lambda, 0\right)$ into \eqref{eq:firm_nonexpansive} and rearranging terms gives 
$$
\begin{aligned}
\left\langle \frac{x}{\lambda}, \prox^{1/\lambda}_{\psi^*}\left(\frac{x}{\lambda}\right) \right\rangle & \amp \geq \amp \left\lVert \prox^{1/\lambda}_{\psi^*}\left(\frac{x}{\lambda}\right) - \prox^{1/\lambda}_{\psi^*}(0) \right\rVert^2 + \left\langle \frac{x}{\lambda}, \prox^{1/\lambda}_{\psi^*}(0) \right\rangle \\
\Rightarrow  \langle x, \nabla \psi^{\lambda}(x) \rangle &\amp \geq \amp \lambda \|\nabla \psi^{\lambda}(x) - \nabla \psi^{\lambda}(0)\|^2 + \left\langle x, \nabla \psi^{\lambda}(0) \right\rangle.
\end{aligned}
$$
Dividing the terms on both sides of the above inequality by $\|x\| \, \|\nabla \psi^{\lambda}(x)\|$,
\begin{align*}
   &\frac{\langle x, \nabla \psi^{\lambda}(x) \rangle}{\lVert x \rVert\lVert \nabla \psi^{\lambda}(x) \rVert} \\
   & \amp \geq \amp 
 \lambda \frac{\|\nabla \psi^{\lambda}(x) - \nabla \psi^{\lambda}(0)\|}{\lVert x \rVert}\cdot \frac{\|\nabla \psi^{\lambda}(x) - \nabla \psi^{\lambda}(0)\|}{\lVert \nabla \psi^{\lambda}(x) \rVert} + \frac{\langle x, \nabla \psi^{\lambda}(0)  \rangle}{\lVert x \rVert\lVert \nabla \psi^{\lambda}(x) \rVert}  \\
 &\amp \geq \amp \lambda \left[\frac{\lVert\nabla\psi^{\lambda}(x) \rVert}{\lVert x \rVert} - \frac{\lVert\nabla\psi^{\lambda}(0) \rVert}{\lVert x \rVert}\right] \left[1 - \frac{\lVert\nabla\psi^{\lambda}(0) \rVert}{\lVert \nabla \psi^{\lambda}(x) \rVert}\right] + \left\langle \frac{x}{\lVert x \rVert}, \frac{\nabla \psi^{\lambda}(0)}{\lVert \nabla \psi^{\lambda}(x) \rVert}\right\rangle\,. \numberthis \label{eq:hmc_proof_main}
\end{align*}
Since $\underset{\lVert x \rVert \rightarrow \infty}{\lim}\; \lVert \nabla \psi^{\lambda}(x) \rVert = \infty$, 
\begin{equation}
\label{eq:hmc_proof_1}
    \underset{\lVert x \rVert \rightarrow \infty}{\lim}\; \left\langle \frac{x}{\lVert x \rVert}, \frac{\nabla \psi^{\lambda}(0)}{\lVert \nabla \psi^{\lambda}(x) \rVert}\right\rangle \amp = \amp 0\,.
\end{equation}
Further, by Assumption~\ref{ass:limsup}
\begin{align*}
    \liminf_{\lVert x \rVert \to \infty} \frac{\lVert\nabla \psi^{\lambda}(x)\rVert}{\lVert x \rVert} & \amp = \amp \frac{1}{\lambda} \liminf_{\lVert x \rVert \to \infty} \frac{\lVert x - \prox^{\lambda}_{\psi}(x)\rVert}{\lVert x \rVert} \\
    & \amp \geq \amp \frac{1}{\lambda} \liminf_{\lVert x \rVert \to \infty} \left(1 - \frac{\|\prox^{\lambda}_{\psi}(x)\|}{\|x\|}\right) \\
    &\amp > \amp 0\,. \numberthis \label{eq:hmc_proof_2}
\end{align*}
Using \eqref{eq:hmc_proof_1} and \eqref{eq:hmc_proof_2} in  \eqref{eq:hmc_proof_main},
\begin{align*}
     &\liminf_{\|x\| \to\infty} \frac{\langle x, \nabla \psi^{\lambda}(x) \rangle}{\lVert x \rVert\lVert \nabla \psi^{\lambda}(x) \rVert} \\
     & \amp \geq \amp \liminf_{\|x\| \to\infty} \left\{ \lambda \left[\frac{\lVert\nabla\psi^{\lambda}(x) \rVert}{\lVert x \rVert} - \frac{\lVert\nabla\psi^{\lambda}(0) \rVert}{\lVert x \rVert}\right] \left[1 - \frac{\lVert\nabla\psi^{\lambda}(0) \rVert}{\lVert \nabla \psi^{\lambda}(x)} \rVert\right] +  \left\langle \frac{x}{\lVert x \rVert}, \frac{\nabla \psi^{\lambda}(0)}{\lVert \nabla \psi^{\lambda}(x) \rVert}\right\rangle \right\} \\
     &\amp \geq \amp \liminf_{\|x\| \to\infty}  \lambda \left[\frac{\lVert\nabla\psi^{\lambda}(x) \rVert}{\lVert x \rVert} - \frac{\lVert\nabla\psi^{\lambda}(0) \rVert}{\lVert x \rVert}\right] \left[1 - \frac{\lVert\nabla\psi^{\lambda}(0) \rVert}{\lVert \nabla \psi^{\lambda}(x) \rVert}\right] + \liminf_{\|x\| \to\infty} \left\langle \frac{x}{\lVert x \rVert}, \frac{\nabla \psi^{\lambda}(0)}{\lVert \nabla \psi^{\lambda}(x) \rVert}\right\rangle\\
     &\amp \geq \amp \liminf_{\|x\| \to\infty}  \lambda \left[\frac{\lVert\nabla\psi^{\lambda}(x) \rVert}{\lVert x \rVert} - \frac{\lVert\nabla\psi^{\lambda}(0) \rVert}{\lVert x \rVert}\right] \lim_{\|x\| \to\infty}\left[1 - \frac{\lVert\nabla\psi^{\lambda}(0) \rVert}{\lVert \lVert \nabla \psi^{\lambda}(x) \rVert \rVert}\right]\\
     &\amp >  \amp 0\,.
\end{align*}
This establishes \eqref{eq:living2}. Next, we show that \eqref{eq:living3also} holds. Note that,
\begin{align*}
    \frac{\|\nabla \psi^{\lambda}(x)\|}{\|x\|} \amp &= \amp \frac{1}{\lambda}\frac{\|x - \prox^{\lambda}_{\psi}(x)\|}{\|x\|} \nonumber \\ 
    \amp &\leq \amp \frac{1}{\lambda}\frac{\|x\| + \|\prox^{\lambda}_{\psi}(x)\|}{\|x\|}\nonumber \\
    \Rightarrow \limsup_{\|x\| \rightarrow \infty} \frac{\|\nabla \psi^{\lambda}(x)\|}{\|x\|} \amp &\leq \amp \limsup_{\|x\| \rightarrow \infty} \frac{1}{\lambda}\left(1 + \frac{\|\prox^{\lambda}_{\psi}(x)\|}{\|x\|} \right) \leq \amp \frac{2}{\lambda} \amp < \amp \infty. 
\end{align*}
Thus, $\limsup_{\|x\| \rightarrow \infty} \|\nabla \psi^{\lambda}(x)\| / \|x\|$ is finite, and \eqref{eq:living3also} holds. By \cite{MR4003576}, the $\pi^{\lambda}$-HMC chain is geometrically ergodic for a sufficiently small $\epsilon$.
\end{proof}
\section{Details for numerical examples}
\subsection{Proximal algorithms}
\label{app:other_algos}

The proximal MCMC algorithms used in the examples are the P-MALA and P-HMC algorithms for the trendfiltering and the nuclear norm examples. For the Bayesian Poisson random effects model, Barker's algorithm is implemented and compared with the proximal versions.

The primary features of the proximal adaptations of the MALA and HMC are the following.
\begin{enumerate}
    \item Use the gradient of the approximated target $\pi^{\lambda}$, i.e. $\nabla \log \pi^{\lambda}(x)$ in the proposal step.
    \item Maintain $\pi$-invariance of the algorithm by employing the Metropolis-Hastings correction step with respect to $\pi$.
\end{enumerate}
Thus, although the proximal MCMC algorithms use the gradient information of the smooth approximations to the target density, they employ the M-H correction step to maintain $\pi$ invariance resulting in samples from $\pi$. The algorithms for proposing samples using P-MALA and P-HMC are given below.

\begin{algorithm}
\caption{P-MALA for $\pi$ given $h > 0$}
\label{algo:P-MALA}
\begin{enumerate}
\item  Given  $X_t = x$, generate proposal $(Y = y) \sim q_{\text{M}}(x, \cdot) \equiv N \left(x - \dfrac{h}{2}\nabla \psi^{\lambda}(x) , h\mathbbm{I}_{d} \right)$
\item  Generate $U \sim U(0,1)$ independently and set
\begin{equation*}
\alpha(x, y) \amp = \amp \min \left\{1, \dfrac{\pi(y)q_{\text{M}}(y, x)}{\pi(x)q_{\text{M}}(x,y)} \right\}
\end{equation*}
\item If $U \leq \alpha(x, y)$, then $X_{t+1} = y$
\item  Else $X_{t+1} = x$
\end{enumerate}
\end{algorithm}

\begin{algorithm} 
\caption{P-HMC for $\pi$ given $\varepsilon > 0$, $L$ a positive integer and $M$ a p.d. matrix}
\label{algo:P-HMC}
    \begin{enumerate}
    \item Given $X_{t} = x$ , $\varepsilon > 0$, and $L \geq 1$
    \item Set $x_{0} = x$, draw $z_{0} \sim N(0, M)$
    \item Use the \emph{leapfrog} one-$\varepsilon$ step equations
    \begin{align*}
    z_{\frac{\varepsilon}{2}} \amp &= \amp z_{0} - \frac{\varepsilon}{2}\nabla \psi^{\lambda}(x_{0}) \\
    x_{\varepsilon} \amp &= \amp x_{0} + \varepsilon M^{-1}z_{\frac{\varepsilon}{2}} \\
    z_{\varepsilon} \amp &= \amp z_{\frac{\varepsilon}{2}} - \frac{\varepsilon}{2}\nabla \psi^{\lambda}(x_{\varepsilon})
\end{align*}
$L$ times sequentially to reach $(x_{0}, z_{0}) \rightarrow (x_{L\varepsilon}, z_{L\varepsilon})$
    \item Draw \(U \sim U(0,1)\) and calculate,
    \begin{equation*}
        \alpha(x_{0}, x_{L\varepsilon}) \amp = \amp \min\left\{1 , e^{-H(x_{L\varepsilon}, z_{L\varepsilon}) + H(x_{0}, z_{0})}\right\}
    \end{equation*}
    where $H(x, z) = \psi(x) + \frac{1}{2}z^{\text{T}}M^{-1}z$ denotes total energy at $(x, z)$
    \item If $U \leq \alpha(x_{0}, x_{L\varepsilon})$, then $X_{t+1} = x_{L\varepsilon}$
    \item Else $X_{t+1} = x$
\end{enumerate}
\end{algorithm}
\noindent The Barker's algorithm due to \cite{livingstone2022barker} uses the gradients of the target in the proposal density,
\begin{equation} \label{eqbark}
    q_{\text{B}}(x, y) \amp = \amp 2.\frac{q_{h}(y-x)}{1 + e^{-(y-x)\Tra\nabla\log\pi(x)}}
\end{equation}
where $q_{h}(\cdot)$ denotes the Gaussian density with variance $h$. Algorithms~\ref{algo:barkerprop} and \ref{algo:Barkers} give a straightforward way to propose from \eqref{eqbark}.

\begin{algorithm}
\caption{Barker proposal on $\mathbb{R}^{d}$}
\label{algo:barkerprop}
 \begin{enumerate}
\item Draw $z \sim q_{h}(.)$
\item Calculate $\phi(x, z) = 1/(1+e^{-z\Tra\nabla\log\pi^{\lambda}(x)})$
\item Set $b(x,z) = 1$ with probability $\phi(x,z)$, and $b(x,z) = -1$ otherwise
\item Set $y = x + b(x,z) \times z$
\end{enumerate}
\end{algorithm}

\begin{algorithm} 
\caption{Metropolis-Hastings with Barker's proposal }
\label{algo:Barkers}
\begin{enumerate}
\item  Given  $X_t = x$, generate proposal $(Y = y) \sim q_{\text{B}}(x, \cdot)$ using Algorithm~\ref{algo:barkerprop}
\item  Generate $U \sim U(0,1)$ independently and set
\begin{equation*}
\alpha(x, y) \amp = \amp \left\{1, \dfrac{\pi(y)}{\pi(x)} \dfrac{q_{\text{B}}(y, x)}{q_{\text{B}}(x,y)}  \right\}
\end{equation*}
\item If $U \leq \alpha(x, y)$, then $X_{t+1} = y$
\item  Else $X_{t+1} = x$
\end{enumerate}
\end{algorithm}

Recall that we employ $\pi^{\lambda}$ invariant Markov chains for drawing samples from the importance sampling proposal.
Algorithm~\ref{algo:MYMALA} presents the $\pi^{\lambda}$-MALA algorithm, and Algorithm~\ref{algo:MYHMC} presents the $\pi^{\lambda}$-HMC algorithm.
\begin{algorithm}
\caption{$\pi^{\lambda}$-MALA given $h > 0$}
\label{algo:MYMALA}
\begin{enumerate}
\item  Given  $X_t = x$, generate $Y = y$ from  $N \left(x - \dfrac{h}{2}\nabla \psi^{\lambda}(x) , h\mathbbm{I}_{d} \right)$ with density $q_{\text{M}}(x, \cdot)$
\item  Generate $U \sim U(0,1)$ independently and set
\begin{eqnarray*}
\alpha(x, y) & = & \min \left\{1, \dfrac{\pi^{\lambda}(y)q_\text{M}(y, x)}{\pi^{\lambda}(x)q_\text{M}(x,y)} \right\}
\end{eqnarray*}
\item If $U \leq \alpha(x, y)$, then $X_{t+1} = y$
\item  Else $X_{t+1} = x$
\end{enumerate}
\end{algorithm}

\begin{algorithm} [t]
\caption{$\pi^{\lambda}$-HMC given $\varepsilon > 0$, $L \geq 1$ and $M_{d \times d} \succ 0$}
\label{algo:MYHMC}
    \begin{enumerate}
    \item Given $X_{t} = x_{0}$, draw $z_{0} \sim N(0, M)$
    \item Use the \emph{leapfrog} one-$\varepsilon$ step equations,
    \begin{eqnarray*}
    z_{\frac{\varepsilon}{2}} & = & z_{0} - \frac{\varepsilon}{2}\nabla \psi^{\lambda}(x_{0}) \\
    x_{\varepsilon} &= &x_{0} + \varepsilon M^{-1}z_{\frac{\varepsilon}{2}} \\
    z_{\varepsilon} & = & z _{\frac{\varepsilon}{2}} - \frac{\varepsilon}{2}\nabla \psi^{\lambda}(x_{\varepsilon})
\end{eqnarray*}
$L$ times sequentially to reach $(x_{0}, z_{0}) \rightarrow (x_{L\varepsilon}, z_{L\varepsilon})$
    \item Draw \(U \sim U(0,1)\) independently, and calculate
    \begin{eqnarray*}
        \alpha(x_{0}, x_{L\varepsilon}) & = & \min\left\{1 , e^{-H(x_{L\varepsilon}, z_{L\varepsilon}) + H(x_{0}, z_{0})}\right\}\,,
    \end{eqnarray*}
    where $H(x, z) = \psi^{\lambda}(x) + \frac{1}{2}z^{\Tra}M^{-1}z$ denotes total energy at $(x, z)$
    \item If $U \leq \alpha(x_{0}, x_{L\varepsilon})$, then $X_{t+1} = x_{L\varepsilon}$
    \item Else $X_{t+1} = x_0$
\end{enumerate}
\end{algorithm}

\subsection{Toy example}
Consider the target distribution with density
\begin{eqnarray}
    \label{eq:laplace}
    \pi_{\beta}(x_1, \dots, x_d) & \propto & \prod_{i=1}^{d} e^{-\psi_{\beta}(x_i)} \amp = \amp \prod_{i=1}^{d} e^{-|x_i|^{\beta}}\,, \qquad x_{i} \in \Real\,,
\end{eqnarray}
$i = 1, 2, \ldots, d$,
for $\beta = 1$ (Laplace) and $\beta = 4$ (super-Gaussian). Then $\psi(x) = \sum_{i=1}^{d}|x_{i}|^{\beta}$, and its Moreau-Yosida envelope is
\begin{eqnarray} \label{eq:MY_env_lap}
    \psi^{\lambda}_{\beta}(x) & = & \min_{y \in \mathbb{R}^{d}}\left(\sum_{i=1}^{d}|y_{i}|^{\beta} + \frac{1}{2\lambda}\|x-y\|^{2}\right)\,.
\end{eqnarray}
Recall that the $i^\text{th}$ component of the proximal mapping of a separable sum of functions is the proximal mapping of the $i^\text{th}$ function in the sum \citep{parikh2014proximal}, i.e., 
\begin{eqnarray*}
\prox^{\lambda}_{\psi_{\beta}}(x_{i}) & = & \arg\min_{y_{i} \in \mathbb{R}}\left(|y_{i}|^{\beta} + \frac{1}{2\lambda}(x_{i}-y_{i})^{2}\right)\,.
\end{eqnarray*}
For $\beta = 1$, we obtain the familiar soft-thresholding function
\begin{eqnarray*}
    \prox^{\lambda}_{\psi_1}(x_{i}) & = & \begin{cases}
        (|x_{i}| - \lambda)\text{sgn}(x_{i}) & \text{if $|x_{i}| \geq \lambda$} \\
        0   & \text{otherwise}\,. 
    \end{cases}
\end{eqnarray*}
\cite{durmus2022proximal} provide the exact expression of $\pi^{\lambda}_{1}$ for $d = 1$:
\begin{eqnarray*}
    \pi^{\lambda}_{1}(x) & = & \dfrac{\exp\left\{ \left(\frac{\lambda}{2} - |x| \right) \mathbbm{1}(|x|\geq \lambda) - \frac{x^2}{2\lambda}\mathbbm{1}(|x| <  \lambda)   \right\}}{2 ( e^{-\lambda/2} + \sqrt{2 \pi \lambda} (\Phi(\sqrt{\lambda}) - 1/2) )}\,.
\end{eqnarray*}
For $\beta = 4$, we obtain
\begin{eqnarray} \label{eq:prox_exp}
    \prox^{\lambda}_{\psi_4}(x_{i}) & = & \dfrac{\sqrt[3]{3}\left(\sqrt{3}\sqrt{\lambda^{3}(27\lambda x_{i}^{2} + 1)} + 9\lambda^{2}x_{i}\right)^{\frac{2}{3}} - 3^{\frac{2}{3}}\lambda}{6\lambda\sqrt[3]{\sqrt{3}\sqrt{\lambda^{3}(27\lambda x_{i}^{2} + 1)} + 9\lambda^{2}x_{i}}} \,,
\end{eqnarray}
which yields an expression of $\pi^{\lambda}_4$ up to a normalization constant. The following figures depict the effect of $\lambda$ on the $\pi^{\lambda}$-MALA Markov chain and the quality of the importance sampling estimator for $d = 1, 10, 20$. For every combination of $\beta$ and $d$, we track the following as $\lambda$ is varied: (i) the estimated $n_e/n$, which reflects the quality of the importance sampling proposal, $\pi^{\lambda}$, (ii) the MCMC effective sample size by $n$ for estimating the mean of $\pi^{\lambda}$, which reflects the mixing quality of the $\pi^{\lambda}$-Markov chain, and (iii) the estimated asymptotic variance of $\hat{\theta}_n^{\text{MY}}$. Each Markov chain is run for $n = 10^6$ iterations. Figures~\ref{fig:laplace} and ~\ref{fig:exp4} show these quantities as a function of $\lambda$ for both $\beta$.
 
These results further strengthen our motivation to recommend choosing $\lambda$ such that $n_e/n$ is in $[0.40, 0.80]$. 
 \begin{figure} 
     \centering
     \includegraphics[scale = .48]{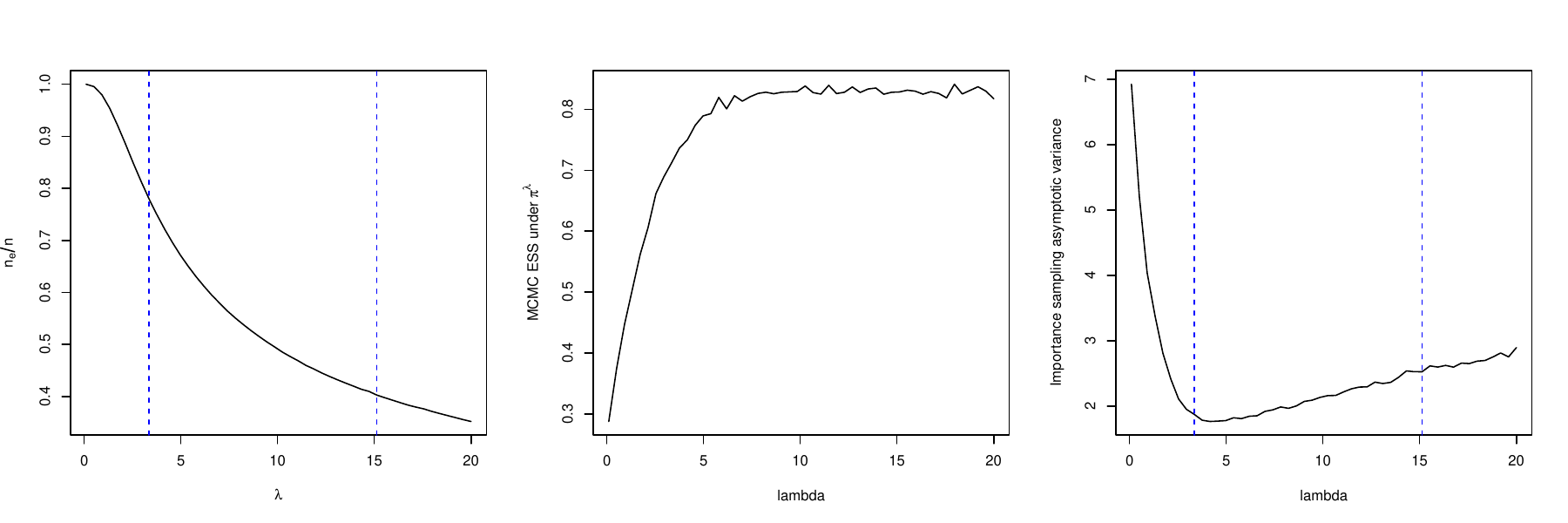}
     \includegraphics[scale = .48]{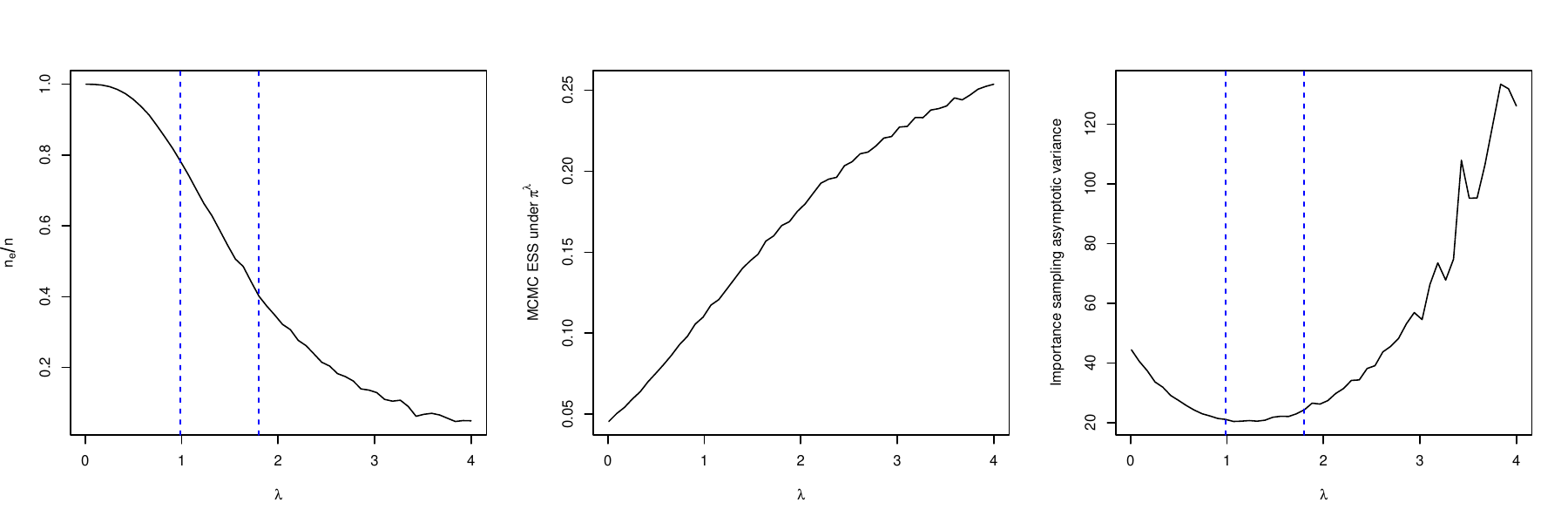}
     \includegraphics[scale = .48]{lap_plot_d20.pdf}
     \caption{(Laplace) From top to bottom $d = 1, 10, 20$. Left column has the importance sampling effective sample size $n_e/n$ for different $\lambda$, middle column has the MCMC effective sample size for the $\pi^{\lambda}$-MALA chain for different $\lambda$, and the right column has the estimated importance sampling asymptotic variance for different values of $\lambda$. The two vertical lines are the values of $\lambda$ that yield $n_e/n \in \{0.40, 0.80\}$.}
     \label{fig:laplace}
 \end{figure}
 \begin{figure} 
     \centering
     \includegraphics[scale = .48]{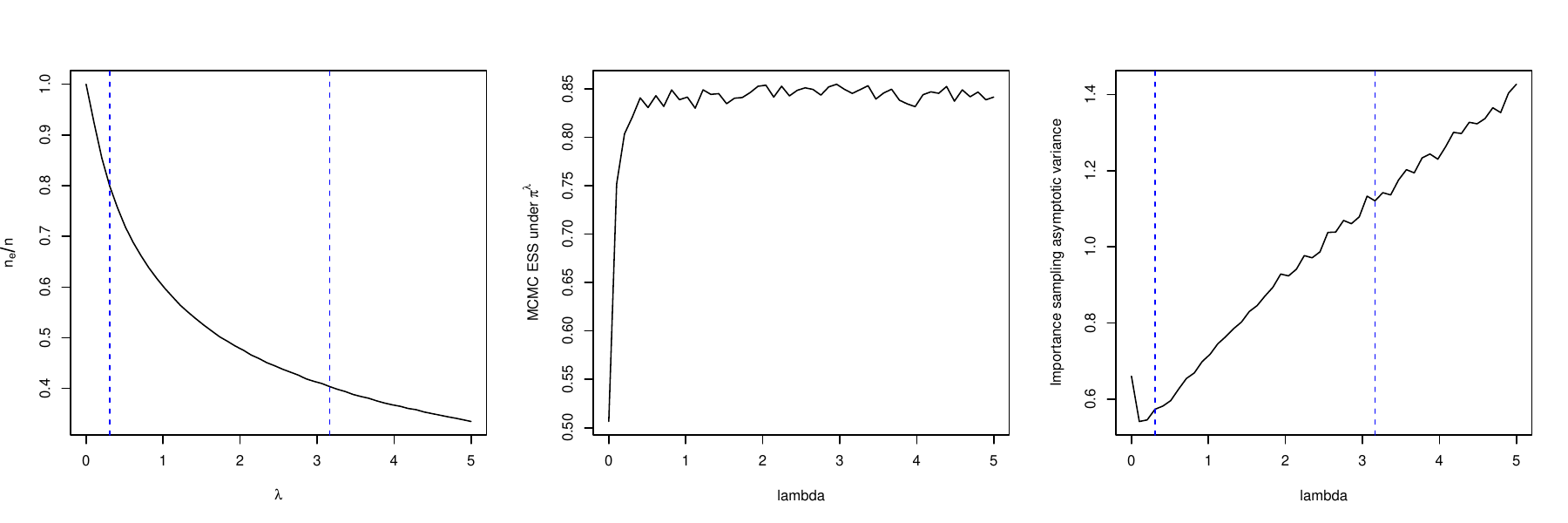}
     \includegraphics[scale = .48]{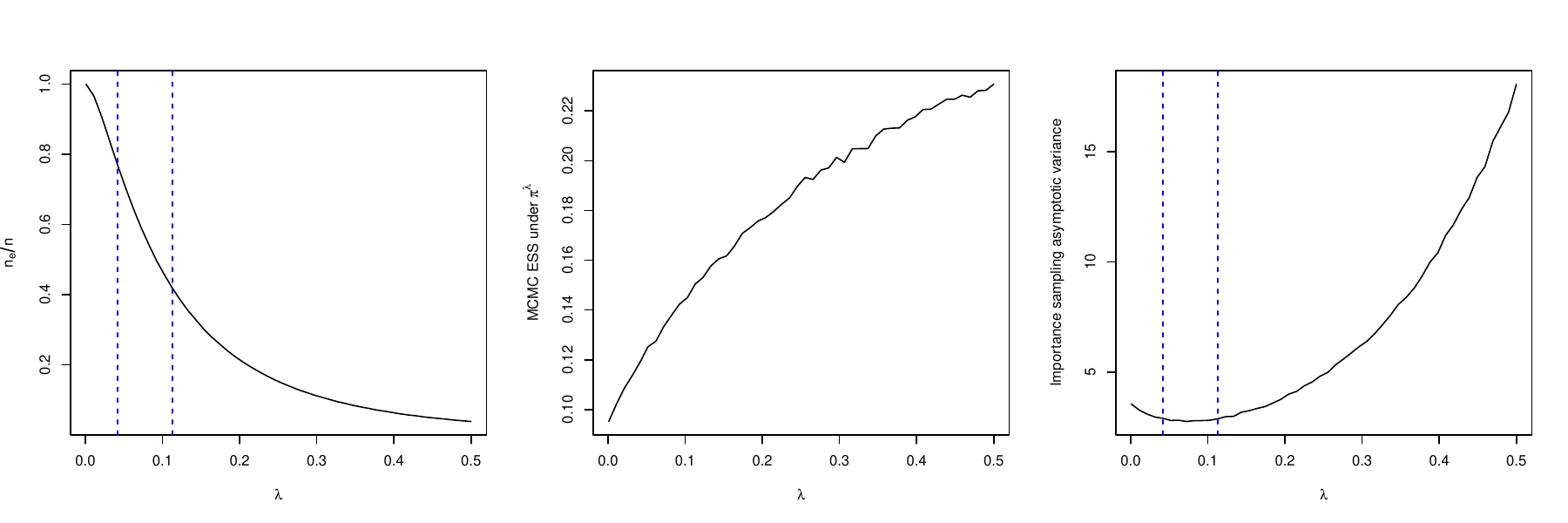}
     \includegraphics[scale = .48]{exp_plot_d20.pdf}
     \caption{(Super-Gaussian) From top to bottom $d = 1,  10, 20$. Left column has the importance sampling effective sample size $n_e/n$ for different $\lambda$, middle column has the MCMC effective sample size for the $\pi^{\lambda}$-MALA chain for different $\lambda$, and the right column has the estimated importance sampling asymptotic variance for different values of $\lambda$. The two vertical lines are the values of $\lambda$ that yield $n_e/n \in \{0.40, 0.80\}$.}
     \label{fig:exp4}
 \end{figure}

\subsection{Bayesian trendfiltering}

The MY envelope of $\psi(\mu)$ is given by
\begin{eqnarray*}
    \psi^{\lambda}(\mu) & = & \underset{\eta \in \Real^{d}}{\min}\; \left\{\psi(\eta) + \frac{1}{2\lambda}\|\mu - \eta\|_{2}^{2} \right\}\\
    & = & \underset{\eta \in \Real^{d}}{\min}\;\left\{ \frac{\|y - \eta\|^{2}_{2}}{2\sigma^{2}} + \frac{1}{2\lambda}\|\mu - \eta\|_{2}^{2} + \alpha\left\|\text{D}_{m}^{(k+1)}\eta\right\|_1\right\},
\end{eqnarray*}
or equivalently,
\begin{equation*}
    \psi^{\lambda}(\mu) \amp = \amp \frac{\left\|y - \prox_{\psi}^{\lambda}(\mu)\right\|^{2}_{2}}{2\sigma^{2}} + \frac{1}{2\lambda}\left\|\mu - \prox_{\psi}^{\lambda}(\mu)\right\|_{2}^{2} + \alpha\left\|\text{D}_{m}^{(k+1)}\prox_{\psi}^{\lambda}(\mu)\right\|_1,
\end{equation*}
where
\begin{equation} 
\label{eq:prox_tf}
    \prox_{\psi}^{\lambda}(\mu) \amp = \amp \underset{\eta \in \Real^{d}}{\arg\min}\;\left\{ \frac{\|y - \eta\|^{2}_{2}}{2\sigma^{2}} + \frac{1}{2\lambda}\|\mu - \eta\|_{2}^{2} + \alpha\|\text{D}_{m}^{(k+1)}\eta\|_1\right\}.
\end{equation}
Completing the square of the quadratic terms in the expression within the braces of \eqref{eq:prox_tf} enables us to express the proximal mapping of $\psi$ as 
\begin{equation} 
\label{eq:prox_tf2}
    \prox_{\psi}^{\lambda}(\mu) \amp = \amp  \underset{\eta \in \Real^{d}}{\arg\min} \left\{ \frac{1}{2}\lVert \eta -  z \rVert_2^2 + \frac{\alpha\sigma^2\lambda}{\sigma^2 + \lambda}\left\lVert \text{D}_{m}^{(k+1)}\eta\right\rVert_1\right\},
\end{equation}
where
\begin{eqnarray*}
z & = &\frac{\sigma^2}{\sigma^2 + \lambda}\mu + \frac{\lambda}{\sigma^2 + \lambda} y.
\end{eqnarray*}
The optimization problem in \eqref{eq:prox_tf2} is solved by an ADMM algorithm implemented in the \texttt{trendfilter} function of the \texttt{glmgen} \texttt{R} package available at \url{https://github.com/statsmaths/glmgen}.

\subsection{Nuclear-norm based low rank matrix estimation}
\label{app:checkerboard}

We consider the $64 \times 64$ checkerboard image of MATLAB, shown in the left panel of Figure~\ref{fig:checkerboard}. The image is low rank as seen in the plot of the singular values of the SVD of the image shown in the right panel of Figure~\ref{fig:checkerboard}.

\begin{figure}
    \centering
    \includegraphics[scale = 0.63]{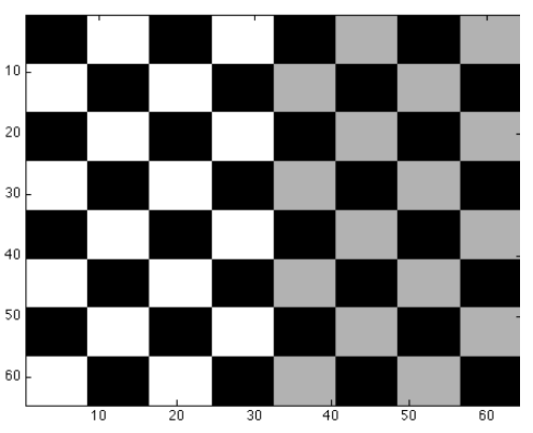}
    \includegraphics[scale = 0.34]{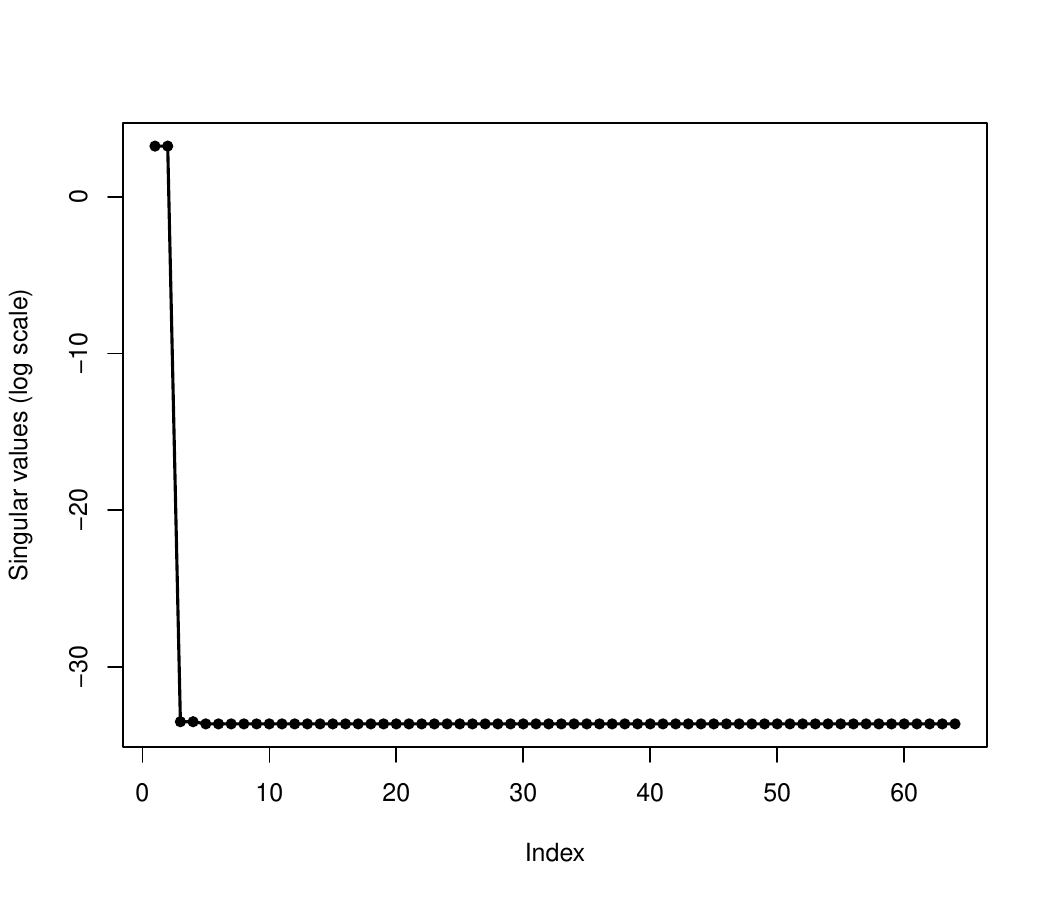}
    \caption{Checkerboard image (\emph{left}), singular values (\emph{right}).}
    \label{fig:checkerboard}
\end{figure}

\subsection{MY envelope for Bayesian Poisson random effects}
\label{app:poisson}
We first derive the joint posterior density function of the Poisson random effects model. The joint prior of $\eta_{i}$'s and $\mu$ is given by,
\begin{align}
    p(\eta_{1}, \eta_{2}, \ldots, \eta_{I}, \mu) \amp &= \amp f(\mu)\prod_{i=1}^{I}f(\eta_{i}|\mu) \nonumber \\
   \amp &\propto \amp \exp\left(-\frac{\mu^{2}}{2c^{2}}\right)\prod_{i=1}^{I} \exp\left(-\frac{1}{2\sigma_{\eta}^{2}}(\eta_{i} - \mu)^{2}\right)  \nonumber\\
   \amp &\propto \amp  \exp\left(-\frac{\sum_{i=1}^{I}(\eta_{i}^{2} + \mu^{2} - 2\eta_{i}\mu)}{2\sigma_{\eta}^{2}}\right)\exp\left(-\frac{\mu^{2}}{2c^{2}}\right)  \nonumber\\
   \amp &\propto \amp  \exp\left[-\left(\frac{\sum_{i=1}^{I}\eta_{i}^2}{2\sigma_{\eta}^{2}} + \frac{\mu^{2}I}{2\sigma_{\eta}^{2}} - \frac{\mu\sum_{i=1}^{I}\eta_{i}}{\sigma_{\eta}^{2}}\right)\right]\exp\left[-\frac{\mu^{2}}{2c^{2}}\right]  \nonumber\\
   \amp &\propto \amp  \exp\left[-\left(\frac{\sum_{i=1}^{I}\eta_{i}^2}{2\sigma_{\eta}^{2}} + \frac{\mu^{2}}{2}\left(\frac{I}{\sigma_{\eta}^{2}} + \frac{1}{c^{2}}\right) - \frac{\mu\sum_{i=1}^{I}\eta_{i}}{\sigma_{\eta}^{2}}\right)\right]\,.  \label{Poisson_prior}
\end{align}
Further, the likelihood function of the observed data is,
\begin{align}
    L(\eta_{1}, \eta_{2}, \ldots, \eta_{I}, \mu | y) \amp &= \amp \prod_{i=1}^{I} \prod_{j=1}^{n_{i}} \frac{\exp(-e^{\eta_{i}}).e^{\eta_{i}y_{ij}}}{y_{ij}!} \nonumber \\
   \amp &= \amp \prod_{i=1}^{I} \prod_{j=1}^{n_{i}} \frac{\exp(\eta_{i}y_{ij} - e^{\eta_{i}})}{y_{ij}!} \nonumber \\
   \amp &= \amp \frac{\exp(\sum_{i=1}^{I}\sum_{j=1}^{n_{i}}(\eta_{i}y_{ij} - e^{\eta_{i}}))}{\prod_{i=1}^{I} \prod_{j=1}y_{ij}!} \nonumber \\
   \amp &= \amp \frac{\exp(\sum_{i=1}^{I}(\eta_{i}\sum_{j=1}^{n_{i}}y_{ij}) - \sum_{i=1}^{I}n_{i}e^{\eta_{i}})}{\prod_{i=1}^{I} \prod_{j=1}y_{ij}!}\,. \label{Poisson_likelihood}
\end{align}
Combining the prior \eqref{Poisson_prior} and likelihood \eqref{Poisson_likelihood} gives us the  posterior
\begin{eqnarray*}
    p(\eta_{1}, \eta_{2}, \ldots, \eta_{I}, \mu | y) & \propto & \exp(-\psi(\eta_{1}, \eta_{2}, \ldots, \eta_{I}, \mu)),
\end{eqnarray*}
where
\begin{equation*}
    \psi(\eta_{1}, \eta_{2}, \ldots, \eta_{I}, \mu) \amp = \amp \frac{\sum_{i=1}^{I}\eta_{i}^2}{2\sigma_{\eta}^{2}} - \frac{\mu\sum_{i=1}^{I}\eta_{i}}{\sigma_{\eta}^{2}} + \sum_{i=1}^{I} n_{i}e^{\eta_{i}} - \sum_{i=1}^{I} \left(\eta_{i}\sum_{j=1}^{n_{i}}y_{ij}\right) + \frac{\mu^{2}}{2}\left(\frac{I}{\sigma_{\eta}^{2}} + \frac{1}{c^{2}}\right).
\end{equation*}
The proximal mapping for $\psi$ does not admit a closed form solution, so we resort to an iterative algorithm using a Newton-Raphson method. This requires evaluating the gradient $\nabla \psi$ and Hessian $\nabla^2\psi$. Denote $u := (\eta_{1}, \eta_{2}, \ldots, \eta_{I}, \mu)$. Then,
\begin{equation*}
    \psi_{\eta_{i}}'(u) \amp := \amp \frac{\partial  \psi(u)}{\partial \eta_{i}} \amp = \amp \frac{\eta_{i}}{\sigma_{\eta}^{2}} - \frac{\mu}{\sigma_{\eta}^{2}} + n_{i}e^{\eta_{i}} - \sum_{j=1}^{n_{i}} y_{ij}
\end{equation*}
for $i = 1, 2, 3, \ldots, I$, and,
\begin{equation*}
    \psi_{\mu}'(u) \amp := \amp \frac{\partial  \psi(u)}{\partial \mu} \amp = \amp \mu \left(\frac{I}{\sigma_{\eta}^{2}} + \frac{1}{c^{2}}\right) - \frac{\sum_{i=1}^{I} \eta_{i}}{\sigma_{\eta}^{2}}\,.
\end{equation*}
Thus, the $(I+1) \times 1$ first order derivative $\nabla \psi(u)$ can be written as,
\begin{align} \label{eq:pois_potential}
   \nabla \psi(u) \amp = \amp \left(
        - \psi_{\eta_{1}}(u), - \psi_{\eta_{2}}(u), \ldots, - \psi_{\eta_{I}}(u), - \psi_{\mu}(u) \right)\Tra.
\end{align}
Since for a given $u$,
\begin{equation*}
    \text{prox}_{\psi}^{\lambda}(u) \amp = \amp \argmin_{v \in \mathbb{R}^{I+1}}\left\{\psi(v) + \frac{\|u - v\|^{2}}{2\lambda} \right\}\amp =: \amp \argmin_{v \in \mathbb{R}^{I+1}} f_{u}^{\lambda}(v),
\end{equation*}
the proximal mapping is the solution of,
\begin{equation} \label{eq:first_der}
    \nabla f_{u}^{\lambda}(v) \amp = \amp \nabla \psi(v) - \frac{u - v}{\lambda} \amp \overset{\text{set}}{=} \amp 0\,.
\end{equation}
Using \eqref{eq:pois_potential} in \eqref{eq:first_der}, we get the proximal solution $\Tilde{v} = (\Tilde{\eta}_{1}, \Tilde{\eta}_{1}, \ldots, \Tilde{\eta}_{I}, \Tilde{\mu})$. In particular,
\begin{align*}
    \Tilde{\mu} \amp = \amp \frac{\sum_{i=1}^{I}\eta_{i}/\sigma_{\eta}^{2} + \mu/\lambda}{(I/\sigma_{\eta}^{2} + 1/c^{2} + 1/\lambda)}\,. 
\end{align*}
Evaluating $\Tilde{v}$ is difficult analytically so we apply the Newton-Raphson algorithm. We require the Hessian of $f_{u}^{\lambda}$ with respect to $(\eta_{1}, \ldots, \eta_{I})$ i.e.,
\begin{equation*}
    \nabla^{2}_{\eta}f_{u}^{\lambda}(v) \amp := \amp \text{diag}(k_{11}, k_{22}, \ldots, k_{II})\,,
\end{equation*}
where $k_{ii} = 1/\sigma_{\eta}^{2} + n_{i}e^{\eta_{i}} + 1/\lambda$, and $\nabla^{2}_{\eta}$ is a map from $\Real^{(I+1) \times 1} \to \Real^{I \times I}$. Let $\eta^{k} = (\eta_{1}^{k}, \eta_{2}^{k}, \ldots, \eta_{I}^{k})$ be the vector of $\eta_{i}$'s at the $k^{\text{th}}$ step. Then the one step Newton-Raphson algorithm to generate $v^{k+1} = (\eta^{k+1}, \mu^{k+1})\Tra$ given $v^{k} = (\eta^{k}, \mu^{k})\Tra$ is given in algorithm \ref{algo:Newton-Raphson}. The algorithm stops when,
\begin{equation*}
    \|\nabla f_{u}^{\lambda}(v^*)\| < \epsilon
\end{equation*}
for a specified tolerance $\epsilon$ and $v^*$ is the approximate minimizer.
\begin{algorithm}[H]
\caption{Newton-Raphson}
\label{algo:Newton-Raphson}
\begin{algorithmic}[1]
\State \vspace*{0.325\baselineskip}Initialize $v^{0} = (\eta^{0}, \mu^{0})\Tra$
\State \vspace*{0.325\baselineskip}$k \gets 0$
\Repeat\vspace*{0.325\baselineskip}
\State  \vspace*{0.325\baselineskip}$\eta^{k+1} = \eta^{k} - \left\{\nabla^{2}_{\eta} f_{u}^{\lambda}\left(v^{k}\right)\right\}^{-1} \nabla_{\eta} f_{u}^{\lambda}\left(v^{k}\right)$
\State \vspace*{0.325\baselineskip} $\mu^{k+1} = \frac{\sum_{i=1}^{I}\eta_{i}^{k+1}/\sigma_{\eta}^{2} + \mu/\lambda}{(I/\sigma_{\eta}^{2} + 1/c^{2} + 1/\lambda)}$
\State \vspace*{0.325\baselineskip} $v^{k+1} = (\eta^{k+1}, \mu^{k+1})\Tra$
\State \vspace*{0.325\baselineskip} $k \gets k + 1$
\Until{$\lVert \nabla f_{u}^{\lambda}(v^k)\rVert < \epsilon$}
\end{algorithmic}
\end{algorithm}
\onehalfspacing
\bibliographystyle{apalike}
\bibliography{ref}

\end{document}